\documentclass[conference]{IEEEtran}
\IEEEoverridecommandlockouts
\usepackage{cite}
\usepackage{amsmath,amssymb,amsfonts}
\usepackage{graphicx}
\usepackage{textcomp}
\usepackage{xcolor}
\usepackage{bm}

\usepackage[ruled]{algorithm2e}
\usepackage{cleveref}
\usepackage{amsthm}
\usepackage{caption}
\usepackage{subcaption}
\usepackage{multirow}
\usepackage{adjustbox}
\usepackage{setspace}

\newtheorem{theorem}{Theorem}

\newtheorem{definition}{Definition}

\newtheorem{example}{Example}

\def\BibTeX{{\rm B\kern-.05em{\sc i\kern-.025em b}\kern-.08em
    T\kern-.1667em\lower.7ex\hbox{E}\kern-.125emX}}

\SetKwRepeat{Do}{do}{while}
    
\begin{document}

\title{DiskANN++: Efficient Page-based Search over Isomorphic Mapped Graph Index using Query-sensitivity Entry Vertex}

\author{
		\IEEEauthorblockN{Jiongkang Ni\textsuperscript{1}, Xiaoliang Xu\textsuperscript{1}, Yuxiang Wang\textsuperscript{1}, Can Li\textsuperscript{1}, Jiajie Yao\textsuperscript{2}, Shihai Xiao\textsuperscript{2}, Xuecang Zhang\textsuperscript{2}}
		\IEEEauthorblockA{\textsuperscript{1} \textit{Hangzhou Dianzi University, Hangzhou, China}, 
			\textsuperscript{2} \textit{Huawei Technologies Co., Ltd, Hangzhou, China}}
		\{hananoyuuki,xxl,lsswyx,lic\}@hdu.edu.cn, \{yaojiajie1,xiaoshihai,zhangxuecang\}@huawei.com
}

\maketitle

\begin{abstract}
Given a vector dataset $\mathcal{X}$ and a query vector $\vec{x}_q$, graph-based Approximate Nearest Neighbor Search (ANNS) aims to build a graph index $G$ and approximately return vectors with minimum distances to $\vec{x}_q$ by searching over $G$. The main drawback of graph-based ANNS is that a graph index would be too large to fit into the memory especially for large-scale $\mathcal{X}$. To solve this, a Product Quantization (PQ)-based hybrid method called DiskANN is proposed to store a low-dimensional PQ index in memory and retain a graph index in SSD, to reduce memory overhead while ensuring a high search accuracy. However, it suffers from two I/O issues that significantly affect the overall efficiency: (1) long routing path from the entry vertex to the query's neighborhood and (2) redundant I/O requests during the routing process. We propose an optimized DiskANN++ to overcome above issues. Specifically, for the first issue, we present a \textit{query-sensitive entry vertex selection strategy} to replace DiskANN's static graph-central entry vertex by a dynamically determined entry vertex that is close to the query. For the second I/O issue, we present an \textit{isomorphic mapping on DiskANN's graph index to optimize the SSD layout} and propose an \textit{asynchronously optimized Pagesearch based on the optimized SSD layout} as an alternative to DiskANN's Beamsearch. Comprehensive experimental studies on real-world datasets demonstrate our DiskANN++'s superiority on efficiency, e.g., we achieve a notable 1.5 X to 2.2 X improvement on QPS compared to DiskANN, given the same accuracy constraint.
\end{abstract}

\section{Introduction}
\label{sec:intro}

Approximate Nearest Neighbor Search (ANNS) \cite{arya1993approximate, indyk1998approximate} has been widely studied recently, which is fundamental for many real-world applications, such as recommendation systems \cite{Wang2022, Sarwar2001}, information retrieval \cite{Xu2022, Flickner1995}, data mining \cite{Adeniyi2016, Bijalwan2014}, and pattern recognition \cite{Cover1967, Zhu2019}. Given a vector dataset $\mathcal{X}$ and a query vector $\vec{x}_q\in \mathcal{X}$, ANNS aims to return approximate nearest neighbors to $\vec{x}_q$ from $\mathcal{X}$ with minimum distance \cite{Wang2021}. ANNS is generally categorized as four types: \textit{hashing-based} \cite{Gionis1999, Gong2020}, \textit{tree-based} \cite{Arora2018, SilpaAnan2008}, \textit{quantization-based} \cite{Jegou2010, Chiu2019}, and \textit{graph-based} \cite{Fu2017, Malkov2016} methods. Among them, graph-based ANNS attracts more attention and has been regarded as the most promising one due to its impressive efficiency and effectiveness over large-scale datasets \cite{Chen2021, Fu2017, Aumueller2020, Wang2021, Shimomura2021, Aoyama2013, Hacid2010, Paredes2005}. 

Graph-based ANNS uses a graph index (e.g., NSG \cite{Fu2017}, HNSW \cite{Malkov2016}, etc.) maintained in the memory to facilitate approximate nearest neighbors retrieval. Each vertex in a graph index represents a vector in $\mathcal{X}$, and an edge between two vertices defines a neighbor relationship. During the retrieval phase, a graph routing over the graph index is conducted, starting from an entry vertex and iteratively exploring the graph index towards the query $\vec{x}_q$ until it converges \cite{Malkov2016}. \textit{The main limitation comes from the heavyweight memory overhead of a graph index.} The memory overhead significantly increases as $\mathcal{X}$ increases, e.g., using HNSW for a billion-scale $\mathcal{X}$ in 128 dimensions would consume around 800 GB memory, which largely exceeds the RAM capacity on a workstation. 


\vspace{0.1cm}
\noindent\textbf{Existing solutions.} Many efforts have been made to alleviate the memory issue of graph-based ANNS for large-scale datasets \cite{Chen2021,Simhadri2022,JayaramSubramanya2019}. The basic idea is to introduce external drives (e.g., SSD) to store high-dimensional vectors, thereby reducing memory overhead. \textit{SPANN} \cite{Chen2021} clusters on the dataset and uses the centroids to construct a spatial partitioning tree (SPT) that resides in memory. High-dimensional vectors within each cluster are stored on SSD. \textit{BBANN} \cite{Simhadri2022} conducts balanced clustering on the dataset and uses the clustered centroids to build a graph index maintained in memory. High-dimensional vectors within each cluster are stored on SSD. The search accuracy of the above methods is significantly affected by the coarse-grained SPT or graph index built on clustered centroids, resulting in a noticeable accuracy loss \cite{Dobson2023}.

\begin{figure*}
	\vspace{-0.5cm}
	\setlength{\abovecaptionskip}{0cm}
    \centering
    \includegraphics[scale=0.53]{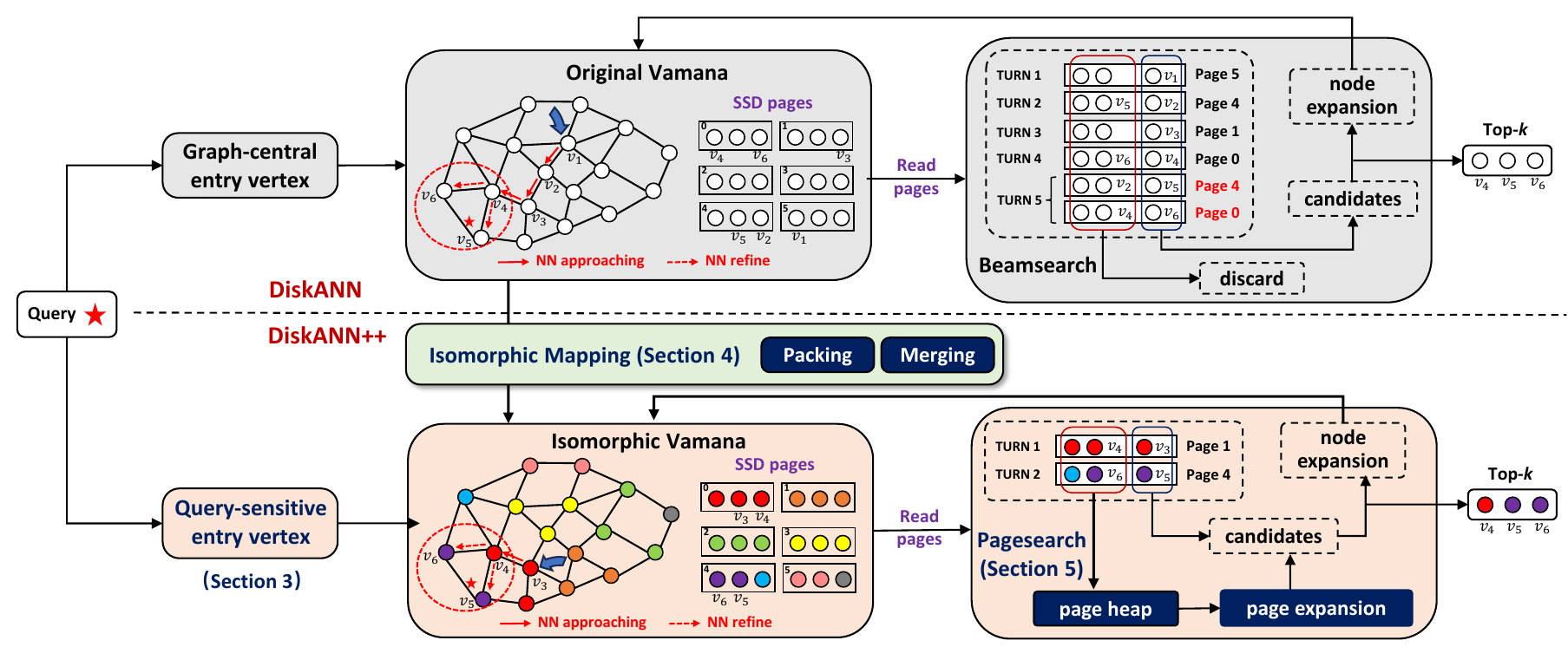}
    \caption{Search examples on DiskANN (top) and DiskANN++ (bottom), given the same query vertex and the Vamana graph index with the same topological structure but different SSD layout.}
    \label{fig:DiskANNPP_overview}
    \vspace{-0.5cm}
\end{figure*}

\vspace{0.1cm}
\noindent\underline{DiskANN \cite{JayaramSubramanya2019}.} This inspires DiskANN, a Product Quantization (PQ)-based hybrid (memory+SSD) method that aims to reduce memory overhead while ensuring a high search accuracy. It first conducts PQ for original high-dimensional vectors to obtain the quantized low-dimensional vectors, then stores a PQ index for quantized vectors in memory and a graph index for original vectors in SSD. PQ index serves for lossy distance calculations, based on which initial candidate neighbors are provided. Graph index serves for re-ranking candidate neighbors based on original vectors, which is a post-verification of lossy distances, so as to enhance the search accuracy. We refer readers to \cite{JayaramSubramanya2019} and \S \ref{sec:pre} for more details. DiskANN has been widely deployed in the industry such as Bing search of Microsoft \cite{Zhang2022a} and many follow-up works present variants of DiskANN to support their own scenarios, e.g., Filter-DiskANN for filtered search \cite{gollapudi2023}, OOD-DiskANN for cross-modal search \cite{Jaiswal2022}, and Fresh-DiskANN for streaming search \cite{Singh2021}. Although DiskANN and its variants achieve a good performance on both memory overhead and search accuracy, they all ignore the important I/O issue (frequent SSD I/O requests for candidates re-ranking) that significantly affect the queries per second (QPS). In general, the latency of accessing SSD is 10X+ greater than that of accessing memory, the more the SSD I/O requests, the more the time for DiskANN. We next discuss the I/O issue of DiskANN.

\vspace{0.1cm}
\noindent\underline{I/O issue of DiskANN.} Figure \ref{fig:DiskANNPP_overview} (top) illustrates the search process of DiskANN on the Vamana graph index. Vertices in Vamana are assigned to SSD pages in a round-robin, e.g., 6 pages for 18 vertices in this example. Given a query vertex (the red star), DiskANN adopts \textit{beamsearch} (discussed in \S \ref{sec:pre}) to retrieve the top-$k$ results. It starts from the static entry vertex (i.e., the central vertex of a graph index) to explore a routing path towards query's neighborhood. We call this the \textit{nearest neighbor approaching} (NN approaching) phase \cite{Xu2021}. DiskANN first reads SSD to get the Page-5 involving the entry vertex $v_1$, then takes $v_1$ as candidate for node expansion and discard the other nodes. In the step of \textit{node expansion}, DiskANN computes the distances between candidates' neighbors to query and requests SSD Page-4 to obtain the closest vertex $v_2$. It repeats until the search locates in the neighborhood of query (indicated by a red dashed circle). NN approaching stops when $v_4$ is explored. Next, DiskANN goes into the \textit{nearest neighbor refine} (NN refine) phase to find the top-$k$ nearest neighbors to query via a nearly brute-force vertex traversal in the query's neighborhood, and both Page-0 and Page-4 are requested to obtain $v_4$'s two neighbors $v_5$, $v_6$. 

The I/O issue of DiskANN is two-fold: (1) \textit{Long routing path in NN approaching phase} and (2) \textit{Redundant I/O requests in NN refine phase}. For (1), since DiskANN takes the graph-central vertex as a static entry vertex for all queries, it would result in a long routing path in NN approaching phase, when the query is far away from the entry vertex. In above example, it needs 3 hops from $v_1$ to query's neighborhood, yielding 4 SSD I/O requests out of the total 6 requests. In practice, queries often arrive randomly, which leads to a large number of long routing paths, so that affecting the overall efficiency. For (2), since DiskANN adopts a random SSD layout, it would result in redundant I/O requests in NN refine phase, when vertices on the same SSD page have less closeness in the graph index. In this example, each accessed page only involves one useful vertex for beamsearch, reducing the data value of each I/O request and leading two redundant I/O requests of Page-0 and Page-4. Actually, they have been accessed in NN approaching phase, however, DiskANN does not know they would be required later and directly discard them as instead.


The aforementioned I/O issue inspires our study in this paper to shorten the routing path and improve the effectiveness of each I/O request (i.e., making an I/O request carry more useful vertices), so that significantly reduces the total number of I/O requests and increases the overall QPS of DiskANN.

\vspace{0.1cm}
\noindent\textbf{Our solution.} We next briefly introduce our solution below.

\noindent\underline{(1) Query-sensitive entry vertex selection (\textbf{\S \ref{sec:entry_vertex}}).} For the long routing path problem, we present a query-sensitive entry vertex selection strategy to dynamically determine the entry vertex in run-time instead of the original static graph-central entry vertex (\textbf{\S \ref{sec:entry_vertex_overview}}). Given a vector dataset $\mathcal{X}$ and a Vamana graph built for $\mathcal{X}$, we first cluster the dataset to acquire $N_{\rm cluster}$ centroids, then we take each centroid as a query to perform ANNS on Vamana to find its top-1 nearest vertex, and finally we add all the $N_{\rm cluster}$ nearest vertices and the graph-central vertex into an \textit{entry vertex candidate list}. For each incoming query $\vec{x}_q$, we linearly scan the candidate list and employ the nearest vertex to $\vec{x}_q$ as the entry vertex. In Figure \ref{fig:DiskANNPP_overview} (bottom), we may select $v_3$ as the entry vertex, leading a 1-hop routing path to query's neighborhood. In \textbf{\S \ref{sec:analysis_entry_vertex}}, we leverage the monotonicity of MSNET \cite{Dearholt1988,Zhu2021} to theoretically prove that using query-sensitive entry vertex would result in a shorter routing path than that of graph-central entry vertex. The shorter routing path usually indicates the less SSD I/O requests.

For redundant I/O issue, we present: \textit{isomorphic mapping of graph index} and \textit{page-level optimization to beamsearch}.

\vspace{0.1cm}
\noindent\underline{(2) Isomorphic mapping of graph index (\textbf{\S \ref{sec:isomorphic}}).} Here, we propose an \textit{isomorphic mapping on Vamana} to optimize the SSD layout. We first apply an injective mapping via \textit{star packing} \cite{Babenko2011} on the original Vamana. It effectively assigns the vertices with great closeness into the same SSD page, so that increasing the data value of each individual I/O request. Then, we apply \textit{bin packing} based on \textit{First Fit Decreasing} (FFD) to make the injective mapping also surjective, i.e., converting the injection to bijection (or isomorphic mapping), thereby ensuring that the original graph's topology and addressing mode are preserved in the new SSD layout. In Figure \ref{fig:DiskANNPP_overview} (bottom), we show the refined SSD layout with colored vertices. Note that, vertices that are close to each other are likely to be assigned to the same SSD page. Suppose we start search from $v_3$ and $v_4$ is the next-hop vertex of routing. Since they are retained in the same SSD page, we only need one SSD I/O request for Page-1 to access them simultaneously. Similarly, we only require one SSD I/O request for Page-4 to access $v_5$ and $v_6$ simultaneously.

\vspace{0.1cm}
\noindent\underline{(3) Page-level optimizations to beamsearch (\textbf{\S \ref{sec:page_based_search}}).} On the basis of the refined SSD layout above, we present a new search algorithm with page-level optimizations called \textit{Pagesearch} as an alternative to beamsearch. The basic idea is to harness the idle CPU resources during SSD I/O requests to further mitigate the search latency. Since we apply isomorphic mapping on Vamana, each page contains more valuable vertices that would be used in the future. So, we design a novel component called \textit{page heap} to asynchronously cache the valuable vertices from previously accessed pages. Next, we expand search from the cached valuable vertices via an asynchronous \textit{page expansion} using the idle CPU resources while waiting the results of SSD I/O requests. The vertices obtained by page expansions would be added into the global candidates for the node expansion of beamsearch. In a nutshell, the page expansion is a complement to node expansion by providing more useful candidates from the previously accessed pages. It is a concurrent operation with SSD I/O requests without introducing extra latency.

\begin{figure}
    \centering
    \includegraphics[width=0.95\linewidth]{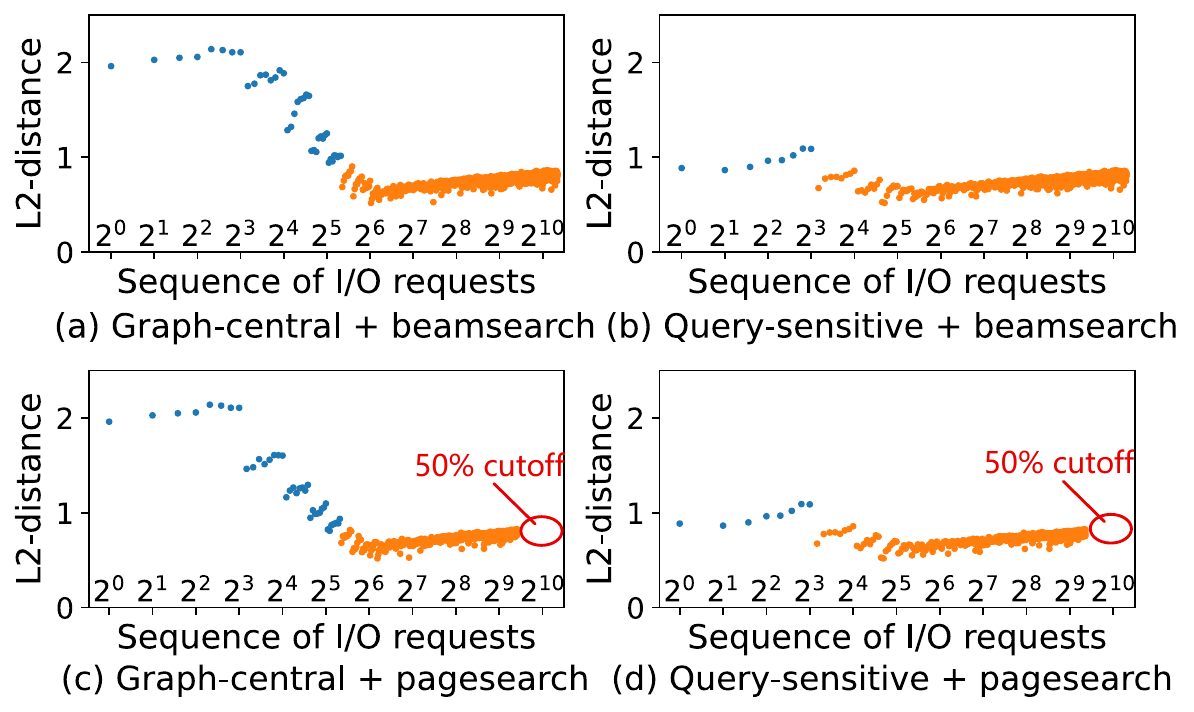}
    \vspace{-0.2cm}
    \caption{Performance comparison on deep100M.}
    \label{fig:two_stage}
    \vspace{-0.8cm}
\end{figure}

Figure \ref{fig:two_stage} shows an ablation analysis of DiskANN using beamsearch and pagesearch with static entry vertex and query-sensitive entry vertex (four combinations in total), given the same queries under the same recall@100 of 99\%. The X-axis is the sequence of I/O requests and Y-axis shows the average vector distance over all accessed vertices in an I/O request to the query. We use the blue (orange) points in plot to represent the I/O requests in NN approaching (NN refine) phase. As more I/O requests are performed, the search is getting closer to the query. Figure \ref{fig:two_stage} (a)-(b) show results using beamsearch with different entry vertex strategies, we found that using query-sensitive entry vertex would significantly reduce the SSD I/O requests in NN approaching phase from 32 to 8. Figure \ref{fig:two_stage} (a)-(c) show results using beamsearch and pagesearch with the same static entry vertex, we found that the SSD I/O requests in NN refine phase are reduced by at least 50\%. Figure \ref{fig:two_stage} (d) shows the results using pagesearch with query-sensitive entry vertex, i.e., our DiskANN++. Comparing with DiskANN (Figure \ref{fig:two_stage} (a)), ours achieves a better I/O efficiency in both two phases, leading to a 2 X improvement on QPS.

\vspace{0.1cm}
\noindent\textbf{Contributions.} Our contributions can be concluded as follows.
\begin{itemize}
    \item We propose a query-sensitive entry vertex selection strategy to determine entry vertex dynamically (\textbf{\S \ref{sec:entry_vertex_overview}}) and prove its effectiveness theoretically (\textbf{\S \ref{sec:analysis_entry_vertex}}).
    \item 
	We present an isomorphic mapping on Vamana to refine the SSD layout (\textbf{\S \ref{sec:pack}}) and analyze its effectiveness using a metric of page compactness (\textbf{\S \ref{sec:connectivity}}). 
    \item We design a novel search algorithm called Pagesearch with page-level optimizations based on the refined SSD layout (\textbf{\S \ref{sec:page_based_search}}), which utilizes idle CPU resources to update global candidates with more valuable vertices.
    \item Extensive experiments (\textbf{\S \ref{sec:evaluation}}) conducted on eight public and commercial datasets with different types and scales, show that our solution achieves a notable QPS improvement ranging from 1.5 X to 2.2 X.
\end{itemize}

\section{Preliminary}
\label{sec:pre}


\begin{definition}
\label{def:ANNS}
\textbf{ANNS \cite{Wang2021}.} Given a vector dataset $\mathcal{X}$, a query vector $\vec{x}_q$, and a parameter $\varepsilon>0$, the goal of ANNS is to find the top-$k$ vectors $\{\vec{x}_1,\cdots,\vec{x}_k\}$ from $\mathcal{X}$ that are approximate nearest neighbors to $\vec{x}_q$. We say a vector $\vec{x}_i\in \mathcal{X}$ is an approximate nearest neighbor to $\vec{x}_q$ if $\delta(\vec{x}_i,\vec{x}_q)\leq (1+\varepsilon)\cdot\delta(\vec{x}_i^*,\vec{x}_q)$, where $\vec{x}_i^*\in \mathcal{X}$ is the $i$-th nearest neighbor vector of $\vec{x}_q$ and $\varepsilon$ is an relaxation parameter controlling the top-$k$ results' quality.
\end{definition}

\begin{definition}
\label{def:graph_index}
\textbf{Graph index.} Given a vector dataset $\mathcal{X}$ and a non-negative distance threshold $\bar{\delta}$, the graph index of $\mathcal{X}$ w.r.t. $\bar{\delta}$ is a graph $G=(V,E)$ with the vertex set $V$ and edge set $E$. (1) There is a bijection $\phi:\mathcal{X}\rightarrow V$, $\forall \vec{x}_v\in \mathcal{X}$, $\exists v\in V$ satisfying $v=\phi(\vec{x}_v)$, i.e., each vertex $v\in V$ corresponds to a vector $\vec{x}_v\in \mathcal{X}$. (2) For any two vertices $v_i,v_j\in V$ ($i\neq j$), there exists an edge $e(v_i,v_j)\in E$ iff $\delta(\vec{x}_{v_i},\vec{x}_{v_j})<\bar{\delta}$.
\end{definition}

\begin{definition}
\label{def:recall}
\textbf{Recall@$k$.} Given a query vector $\vec{x}_q$, $\mathcal{R}^*$ records the ground-truth nearest neighbors with $k$ vectors from $\mathcal{X}$ and $\mathcal{R}$ records $k$ approximate nearest neighbors returned by ANNS. Then we define the Recall@$k$ as follows.
\begin{equation}
\label{eq:recall}
\text{Recall}\text{@}k = \frac{\left|\mathcal{R}^* \cap \mathcal{R}\right|}{|\mathcal{R}^*|} = \frac{\left|\mathcal{R}^* \cap \mathcal{R}\right|}{k}
\end{equation}
\end{definition}

\begin{definition}
\label{def:qps}
\textbf{Queries Per Second (QPS).} QPS is a metric indicating the number of queries that an ANNS method can handle per second. Suppose an ANNS method processes $N_q$ queries within$T$ seconds, then we have $\text{QPS}=N_q/T$.
\end{definition}


\vspace{0.1cm}
\noindent\textbf{Briefly introduction to DiskANN.} We briefly introduce the SSD layout of the graph index and beamsearch used in DiskANN, which is important for understanding our solution.

\vspace{0.1cm}
\noindent\underline{Original SSD layout.} DiskANN employs a straightforward method to store the graph index $G=(V,E)$ in SSD. 

\begin{definition}
\label{def:data_blocks}
\textbf{Data block.} Given a vertex $v\in V$, we define the data block of $v$ as $b_v=\langle\vec{x}_v,N(v)\rangle$, which is the basic unit for SSD storage. (1) $\vec{x}_v\in \mathcal{X}$ is the vector of $v$. (2) $N(v)$ recodes the identities of $v$'s neighbors in $G$. (3) We use $v$ to indicate the identity of the data block $b_v$, denoted by $b_v.{\sf ID}=v$.
\end{definition}

DiskANN stores the data blocks of all $|V|$ vertices to SSD pages in a round-robin with page alignment. We use $L=\{P_1,\cdots, P_n\}$ to denote the SSD layout with $n$ pages and each $P=\{b_{v_1},\cdots,b_{v_b}\}$ contains $b$ data blocks. For simplicity, in the rest of this paper, we use $L.{\sf IDs}=\{P_1.{\sf IDs},\cdots, P_n.{\sf IDs}\}$ to denote the logic view of a layout that only consists of the identities of all data blocks, where $P.{\sf IDs}=\{v_1,\cdots,v_b\}$.

\setlength{\textfloatsep}{0.1cm}
\begin{algorithm}[t]
\small
\setstretch{0.95}
\caption{Beamsearch($G$, $\vec{x}_q$, $B$, $L_s$, $k$)}
\label{alg:beamsearch}
\LinesNumbered
\KwIn{$G$, $\vec{x}_q$, $B$, $L_s$, $k$ of top-$k$}
\KwOut{approximate nearest top-$k$ neighbors $\mathcal{R}$ to $\vec{x}_q$}
\tcp{\footnotesize Initialization: lines 1-3}
$v_e \leftarrow$ central vertex of $G$ \tcc*[r]{{\scriptsize static entry vertex}}
$\mathcal{C} \leftarrow \{v_e\}$ \tcc*[r]{{\scriptsize candidates initialized as $\{v_e\}$}}
$\mathcal{R} \leftarrow \emptyset$ \tcc*[r]{{\scriptsize top-$k$ results initialized as $\emptyset$}}
\Do{$\mathcal{F}\neq\emptyset$}{
$\mathcal{F} \leftarrow$ top-$B$ unvisited vertices from $C$ for expansion \;
    \tcp{\footnotesize prepare SSD I/O requests: line 6-10}
    $\mathcal{P} \leftarrow \emptyset$ \tcc*[r]{{\scriptsize page placeholders}}
    \For{$v_i \in \mathcal{F}$}{
        $P_j \leftarrow$ register read for page containing $b_{v_i}$ \;
        $\mathcal{P} \leftarrow \mathcal{P} \cup P_j$ \;
    }
    \tcp{\footnotesize SSD I/O requests: line 11}
    read all required pages in $\mathcal{P}$ from SSD \;
    \tcp{\footnotesize Node expansion: lines 12-15}
    \For{$v_i \in \mathcal{F}$}{
        $b_{v_i}=\langle \vec{x}_{v_i},N(v_i)\rangle \leftarrow$ obtained from $\mathcal{P}$ \;
        NeighborExpansion($\vec{x}_{q}$, $b_{v_i},\mathcal{C},\mathcal{R},L_s,k$) \;
    }
}
\Return $\mathcal{R}$\;
\end{algorithm}

\setlength{\textfloatsep}{0.1cm}
\setlength{\floatsep}{0.1cm}
\begin{algorithm}[t]
\setstretch{0.95}
\small
\caption{NeighborExpansion($\vec{x}_{q}$, $b_v$, $\mathcal{C}$, $\mathcal{R}$, $L_s$, $k$)}
\label{alg:neighbor_expand}
\LinesNumbered
\SetNoFillComment
\KwIn{$\vec{x}_{q}$, $b_v=\langle \vec{x}_v,N(v) \rangle$, $\mathcal{C}$,$\mathcal{R}$, $L_s$, $k$ of top-$k$}
$\mathcal{C} \leftarrow \mathcal{C}\cup N(v)$ \tcc*[r]{{\scriptsize sort $\mathcal{C}$ by PQ distance to $\vec{x}_q$}}
\While{$|\mathcal{C}| > L_s$}{
    pop back from $\mathcal{C}$ \;
}
$\mathcal{R} \leftarrow \mathcal{R}\cup \{\vec{x}_v\}$ \tcc*[r]{{\scriptsize sort $\mathcal{R}$ by full distance to $\vec{x}_q$}}
\While{$|\mathcal{R}| > k$}{
    pop back from $\mathcal{R}$ \;
}
\end{algorithm}

\vspace{0.1cm}
\noindent\underline{Beamsearch.} Beamsearch is the core search algorithm of DiskANN, but the details are ignored in \cite{JayaramSubramanya2019}. Given an SSD-resident graph index $G=(V,E)$ of the vector dataset $\mathcal{X}$, query vector $\vec{x}_q$, beam size $B$, search width $L_s$, and the size $k$ of top-$k$ results, \Cref{alg:beamsearch} shows the procedure of beamsearch. (1) It starts search from a static graph-central entry vertex $v_e\in V$ with a candidate set $\mathcal{C}=\{v_e\}$ and an empty top-$k$ results $\mathcal{R}$ (lines 1-3). It's worth mentioning that vertices in $\mathcal{C}$ are ranked in ascending order of their PQ distances to $\vec{x}_q$ using the memory-resident quantized vectors. While the vertices in $\mathcal{R}$ are ranked in ascending order of their full distances to $\vec{x}_q$ using the SSD-resident original vectors. (2) It chooses the unvisited top-$B$ vertices (at most $B$) from $\mathcal{C}$, denoted by $\mathcal{F}$ and creates a page placeholder $P_j$ into $\mathcal{P}$ to register read request (line 6-10). Finally, it reads all required pages from SSD (line 11). (3) It uses all vertices from $\mathcal{F}$ to perform node expansion (lines 12-15). The node expansion phase mainly consists of $B$ times NeighborExpansion (Algorithm \ref{alg:neighbor_expand}). NeighborExpansion uses a data block $b_v=\langle \vec{x}_v,N(v) \rangle$ to update $\mathcal{C}$ with the neighbors $N(v)$ in ascending order of PQ distance to $\vec{x}_q$ and ensure $\mathcal{C}$'s length $\leq L_s$. Then, it updates the top-$k$ results $\mathcal{R}$ with $\vec{x}_v$ in ascending order of full distance to $\vec{x}_q$ and ensure $\mathcal{R}$ 's length $\leq k$. The re-ranking operation to $\mathcal{R}$ is the key to guarantee the search accuracy. (4) It repeats above until no new vertex is visited and returns $\mathcal{R}$ as the top-$k$ results.

\vspace{0.1cm}
\noindent\textbf{Problem definition.} Given a vector dataset $\mathcal{X}$, a memory constraint $M$, a graph index $G$ built for $\mathcal{X}$, and a query vector $\vec{x}_q$, DiskANN aims to maintain a PQ index within $M$ size in memory and leverage the SSD-resident $G$ to return the top-$k$ approximate nearest neighbors to $\vec{x}_q$ with a high Recall@$k$. 

\vspace{0.1cm}
\noindent\underline{Our goal.} On this basis, our goal in this paper is: Given the conditions unchanged, we aim to design a refined SSD layout of graph index and a new search algorithm based on such layout, to retrieve the top-$k$ results having the same Recall@$k$ as DiskANN, while improving the QPS by reducing the SSD I/O requests, i.e., increasing QPS without sacrificing accuracy.

\section{Query-sensitive entry vertex}
\label{sec:entry_vertex}
To solve the first I/O issue of DiskANN, we present a query-sensitive entry vertex selection strategy discussed in \S \ref{sec:entry_vertex_overview}.

\subsection{Method Overview}
\label{sec:entry_vertex_overview}

\noindent\textbf{Offline candidate entry vertices generation.} Given a vector dataset $\mathcal{X}$ and the Vamana graph index $G$ of $\mathcal{X}$, we generate the entry vertex candidates as follows. First, we employ the \textit{mini-batch-kmeans} \cite{Peng2018} to cluster $\mathcal{X}$ into $N_{\rm cluster}$ clusters $\{c_1,\cdots,c_{N_{\rm cluster}}\}$. Second, we take each $c_i$'s centroid as an input query $\vec{x}_q$ to find its top-$1$ nearest neighbor from $G$. Finally, we record all the nearest neighbors for $N_{\rm cluster}$ centroids in a linear table as the entry vertex candidates. 

\vspace{0.1cm}
\noindent\textbf{Online entry vertex selection. } Since each candidate is a representative vertex of a partition of $\mathcal{X}$, it's reasonable to take it as the entry vertex when the incoming query locates in the same partition. To this end, given the candidate entry vertices and a query vertex $\vec{x}_q$, we select a candidate entry vertex with the closest distance to $\vec{x}_q$ as the entry vertex.

\subsection{Theoretical Analysis of Effectiveness}
\label{sec:analysis_entry_vertex}
We next theoretically prove that using query-sensitivity entry vertex can provide a tighter upper bound on the routing length than that of graph-central entry vertex. Since the graph index of DiskANN is developed based on Monotonic Search Network (MSNET) \cite{JayaramSubramanya2019}, similar to \cite{Fu2017}, we leverage MSNET's monotonicity to complete our analysis.

\begin{definition}
\label{def:mp}
\textbf{Monotonic Path}. Given a graph index $G(V,E)$ for a vector dataset $\mathcal{X}$. Let $v_s,v_t\in V$, an $l$-hop path $\mathcal{P}(v_s,v_t)$ from $v_s$ to $v_t$ is a Monotonic Path, iff $\exists v_1,\cdots,v_{l+1}\in V$ ($v_1=v_s, v_{l+1}=v_t$) satisfying: (1) each pair of adjacent vertices in $\mathcal{P}$ has an edge $e(v_i,v_{i+1})\in E$, (2) $\delta(\phi^{-1}(v_t),\phi^{-1}(v_{i+1}))<\delta(\phi^{-1}(v_t),\phi^{-1}(v_{i}))$, where $\phi^{-1}(v)$ denotes $v$'s vector $\vec{x}_v\in \mathcal{X}$ (Definition \ref{def:graph_index}). This implies that the greater the number of hops (in a path) from $v_i$ to $v_t$, the larger the vector distance between them. We call such a path a monotonic path $\mathcal{MP}(v_s,v_t)$.
\end{definition}

\begin{definition}
\label{def:msnet}
\textbf{Monotonic Search Network}. Given a graph index $G(V,E)$ for a vector dataset $\mathcal{X}$. $G$ is a Monotonic Search Network, iff $\forall v_s,v_t\in V, \exists \mathcal{MP}(v_s,v_t)$ on $G$.
\end{definition}

We provide a theorem showing the relation between the entry vertex and the upper bound on a routing path's length. 

\begin{theorem}
\label{th:1}
Given a vector dataset $\mathcal{X}$ distributed in a unit sphere $\mathcal{B}$ in $\mathbb{R}^d$, a MSNET $G(V,E)$, and a query $\vec{x}_q\in \mathcal{B}$. For query-sensitivity entry vertex, we have $N_{\rm cluster}$ entry vertex candidates $\{v_{c_1},\cdots,v_{c_{N_{\rm cluster}}}\}$, where $v_{c_j}$ is the top-$1$ nearest neighbor to the centroid of cluster $c_j$. For the static entry vertex used in DiskANN, we use $v_{c_0}$ to denote the graph-central vertex. $\vec{x}_{q^*}$ is the top-$1$ nearest neighbor of $\vec{x}_q$ within $\mathcal{X}$ ($v_{q^*}=\phi(\vec{x}_{q^*})$ is the vertex in $G$). Given above assumptions, the following inequality holds for $\exists j\in \{1,\cdots,N_{\rm cluster}\}$:
\begin{equation}
    \label{eq:1}
    \overline{\left| \mathcal{MP}(v_{c_j},v_{q^*}) \right|} \le \overline{\left| \mathcal{MP}(v_{c_0},v_{q^*}) \right|}\quad ,
\end{equation}
where $\overline{\left| \mathcal{MP}(\cdot,\cdot) \right|}$ represents the upper bound on the length of a Monotonic Path, i.e., the number of hops.
\end{theorem}

\begin{proof}
\label{proof:th1}
We prove this theorem with the following two cases. 
\noindent\textbf{Case 1.} Suppose $N_{\rm cluster}=1$, then we have $v_{c_1}=v_{c_0}$. This is because $c_1$ is the entire graph $G$ and the central vertex of $G$ actually is the top-$1$ nearest neighbor of $c_1$'s centroid. Thus, we have $\overline{\left| \mathcal{MP}(v_{c_1},v_{q^*}) \right|}=\overline{\left| \mathcal{MP}(v_{c_0},v_{q^*}) \right|}$.

\vspace{0.1cm}
\noindent\textbf{Case 2.} For $N_{\rm cluster}>1$, we introduce the concept of \textit{open sphere} to derive the upper bound on a routing path's length.

Let $\mathcal{H}(\vec{x}_{q^*},\theta)$ denotes an open sphere in $\mathbb{R}^d$ with center $\vec{x}_{q^*}$ and radius $\theta$, $\mathcal{H}_{\rm vol}(\vec{x}_{q^*},\theta)$ denotes the volume of $\mathcal{H}(\vec{x}_{q^*},\theta)$. Considering a monotonic path $\mathcal{MP}(v_{c_j},v_{q^*})$ involving vertices $\left\{v_1,\cdots,v_{l+1},v_{q^*} \right\}$ ($v_1=v_{c_j}$ is an entry vertex), for simplicity, we use the notation ${\rm Vol}_{i}$ to denote the volume of a sphere, i.e., ${\rm Vol}_{i}=\mathcal{H}_{\rm vol}(\vec{x}_{q^*},\delta(\vec{x}_{q^*},\vec{x}_{v_i}))$, for $i\in \{1,2,...,l\}$. Since the volume of a sphere is calculated as
\begin{equation}
        \begin{aligned}
            \mathcal{H}_{\rm vol}(\cdot,\theta)=\frac{(\sqrt{\pi}\theta)^d}{\Gamma(\frac{d}{2})+1}\quad ,
        \end{aligned} 
\end{equation}
then we have the following for $i\in \{1,2,...,l\}$:
\begin{equation}
        \label{eq:eq2}
        \begin{aligned}
        \frac{\mathit{{\rm Vol}}_{i+1}}{\mathit{{\rm Vol}}_i}=(\frac{\delta(\vec{x}_{q^*},\vec{x}_{v_{i+1}})}{\delta(\vec{x}_{q^*},\vec{x}_{v_i})})^d \quad .
        \end{aligned}
\end{equation}

We next show the relationship between $\delta(\vec{x}_{q^*},\vec{x}_{v_{i+1}})$ and $\delta(\vec{x}_{q^*},\vec{x}_{v_i})$. Given a cluster $c_j$, its diameter is $\overline{R}={\max}(\delta(\vec{x}_u,\vec{x}_v))$ for $\vec{x}_u, \vec{x}_v \in c_j$. We use $R^*={\min}\left| \delta(\vec{x}_{q^*},\vec{x}_{v_i})-\delta(\vec{x}_{q^*},\vec{x}_{v_j}) \right|$ to represent the minimum distance difference in any monotonic path $\mathcal{MP}(v_s,v_{q^*})$ from an arbitrary vertex $v_s$ to $v_{q^*}$, where $v_i,v_j\in \mathcal{MP}(v_s,v_{q^*})$ for $i\neq j$. According to Definition \ref{def:mp}, we have $\delta(\vec{x}_{q^*},\vec{x}_{v_i})>\delta(\vec{x}_{q^*},\vec{x}_{v_{i+1}})$, so that $\delta(\vec{x}_{q^*},\vec{x}_{v_i})-\delta(\vec{x}_{q^*},\vec{x}_{v_{i+1}})\geq R^*>0$ naturally holds. Thereby, due to $\overline{R}\geq \delta(\vec{x}_{q^*},\vec{x}_{v_i})\geq \delta(\vec{x}_{q^*},\vec{x}_{v_i})-\delta(\vec{x}_{q^*},\vec{x}_{v_{i+1}})\geq R^*$, we have
    \begin{equation}
    \small
        \label{eq:eq3}
        \begin{aligned}
         \overline{R}\cdot (\delta(\vec{x}_{q^*},\vec{x}_{v_i})-\delta(\vec{x}_{q^*},\vec{x}_{v_{i+1}}))&\geq \delta(\vec{x}_{q^*},\vec{x}_{v_i})\cdot R^* \\ 
        \Rightarrow \frac{\overline{R}-R^*}{\overline{R}}&\geq \frac{\delta(\vec{x}_{q^*},\vec{x}_{v_{i+1}})}{\delta(\vec{x}_{q^*},\vec{x}_{v_i})}\quad .
        \end{aligned}
 \end{equation}
    
According to Eq. \ref{eq:eq2} and Eq. \ref{eq:eq3}, we have 
\begin{equation}
\small
        \begin{aligned}
        \frac{\mathit{{\rm Vol}}_{i+1}}{\mathit{{\rm Vol}}_i}\le (\frac{\overline{R}-R^*}{\overline{R}})^d\quad .
        \end{aligned}
 \end{equation}
 
Given $\overline{\mathit{{\rm Vol}}}=\mathcal{H}_{\rm vol}(\vec{x}_{q^*},\overline{R})$, we have:
\begin{equation}
\small
        \begin{aligned}
            \mathit{{\rm Vol}}_{l+1}               &\le \mathit{{\rm Vol}}_{l}\cdot(\frac{\overline{R}-R^*}{\overline{R}})^d \\
                                             &\le \mathit{{\rm Vol}}_1\cdot(\frac{\overline{R}-R^*}{\overline{R}})^{ld} \le \overline{\mathit{{\rm Vol}}}\cdot(\frac{\overline{R}-R^*}{\overline{R}})^{ld} \\
            \Rightarrow \log(\mathit{{\rm Vol}}_{l+1}) &\le \log(\overline{\mathit{{\rm Vol}}})+ld\cdot \log(\frac{\overline{R}-R^*}{\overline{R}}) \\
            \Rightarrow l                    &\le \frac{
                                                        \log(\delta(\vec{x}_{q^*},\vec{x}_{v_{l+1}}))-\log(\overline{R})
                                                    }{
                                                        \log(\overline{R}-R^*)-\log(\overline{R})
                                                    } \triangleq f(\overline{R})\quad .
        \end{aligned}
\end{equation}

Above derivation shows the upper bound on the length of $l$ is $f(\overline{R})$, we next analyze the monotonicity of $f(\overline{R})$ as
    \begin{equation}
    \small
        \begin{aligned}
            \frac{df}{d\overline{R}}
            =& \frac{
                \frac{1}{\overline{R}}[\log\overline{R}-\log(\overline{R}-R^*)]
            }{
                (\log(\overline{R}-R^*)-\log(\overline{R}))^2
            }\\
            +&
            \frac{
                (\log\overline{R}-\log\delta(\vec{x}_{q^*},\vec{x}_{v_{l+1}}))(\frac{1}{\overline{R}-R^*}-\frac{1}{\overline{R}})
            }{
                (\log(\overline{R}-R^*)-\log(\overline{R}))^2
            }\quad .
        \end{aligned}
    \end{equation}
    Since we have the following inequalities:
    \begin{itemize}
        \item $\log\overline{R}-\log(\overline{R}-R^*)>0$\quad ,
        \item $\log\overline{R}-\log\delta(\vec{x}_{q^*},\vec{x}_{v_{l+1}})>0$\quad ,
        \item $\frac{1}{\overline{R}-R^*}-\frac{1}{\overline{R}}>0$\quad ,
    \end{itemize}
    $f(\overline{R})$ is monotonically increasing. Moreover, since we assume $N_{\rm cluster}>1$, we have $\overline{R} \leq \frac{1}{2}$ and the following holds.
\end{proof}
\vspace{-0.7cm}
\begin{equation}
    \small
       \overline{\left| \mathcal{MP}(v_{c_j},v_{q^*}) \right|}= f(\overline{R})+1\le f(\frac{1}{2})+1=\overline{\left| \mathcal{MP}(v_{c_0},v_{q^*}) \right|}
\end{equation}

\subsection{Complexity Analysis}
\label{sec:complexity_entry_vertex}
\noindent\textbf{Complexity.} The time complexity of offline candidates generation stems from the mini-batch-kmean, which is $O(r\cdot N_{\rm batch}\cdot N_{\rm cluster}\cdot d)$ \cite{Peng2018}, where $r$ is the number of clustering iterations, $N_{\rm batch}$ is the mini-batch size, $N_{\rm cluster}$ is the number of entry vertices, and $d$ is the dimensionality of a vector dominating the distance calculation's cost. The time complexity of online entry vertex selection arises from the distance calculations between each candidate and $\vec{x}_q$, which is $O(N_{\rm cluster}\cdot d$).

\vspace{0.1cm}
\noindent\textbf{Remarks.} (1) The efficiency improvement brought by query-sensitivity entry vertex comes from the trade-off between the time spent for online entry vertex selection (in memory) and the time saved for SSD I/O. Since reading an SSD page is at least 10 X slower than that of reading memory, it's natural that we have a good efficiency compared with DiskANN. (2) In practical, we prefer a smaller $N_{\rm cluster}$ to decrease the overhead of query-sensitivity entry vertex selection, when SSD's I/O bandwidth is large (or we say that SSD I/O is fast). Otherwise, a larger $N_{\rm cluster}$ would be better. We experimentally show $N_{\rm clsuter}$'s effect under diverse SSD I/O bandwidth in \S \ref{sec:parameter}.

\section{Isomorphic Mapping on Vamana}
\label{sec:isomorphic}
In this section, we present an isomorphic mapping on Vamana to refine its SSD layout, thus increasing the data value of each SSD page (\S \ref{sec:pack}), then analyze its effectiveness via \textit{algebraic connectivity} \cite{Fiedler1973} (\S \ref{sec:connectivity}), and show its complexity (\S \ref{sec:complexity_mapping}). In \S \ref{sec:pagesearch}, we will discuss how to use such a refined SSD layout to accelerate search via a page-based search.


\subsection{Pack-Merge-based Method}
\label{sec:pack}
Given the page capacity of $b$, DiskANN writes data blocks to SSD pages in a round-robin (Figure \ref{fig:order_step} (a)). In this way, we can quickly compute a given vertex's resident page, yielding an efficient SSD addressing. For example, given a vertex $v_i$, $b_{v_i}$ is the $i\%b$-th data block of the $\lceil\frac{i}{b}\rceil$-th SSD page. 

We aim to refine the SSD layout by increasing the data value of each SSD page while keeping the original addressing mode unchanged. Achieving this is non-trivial. \textit{Edge-based graph partitioning} methods \cite{Zhang2017,Xie2014,Petroni2015} suffer from the redundancy in vertices, e.g., one vertex would appear multiple times in different SSD pages, invalidating original addressing mode. \textit{Vertex-based graph partitioning} methods \cite{Karypis1970,Stanton2012,Tsourakakis2014} fail to ensure that each SSD page contains the same number of vertices, also undermining the calculation of SSD offsets. \textit{Graph reordering} method \cite{Wei2016} keeps the addressing mode unchanged, but they load the entire graph and construct a reverse index in memory render them inapplicable to large-scale graphs. We compare our solution to representative reordering methods in \S \ref{sec:evaluation}.

Different from them, we present a low-memory overhead, low-time complexity lightweight isomorphic mapping on Vamana, while retaining the original rapid SSD addressing mode.

\begin{figure}
	\vspace{-0.2cm}
    \centering
    \includegraphics[scale=0.44]{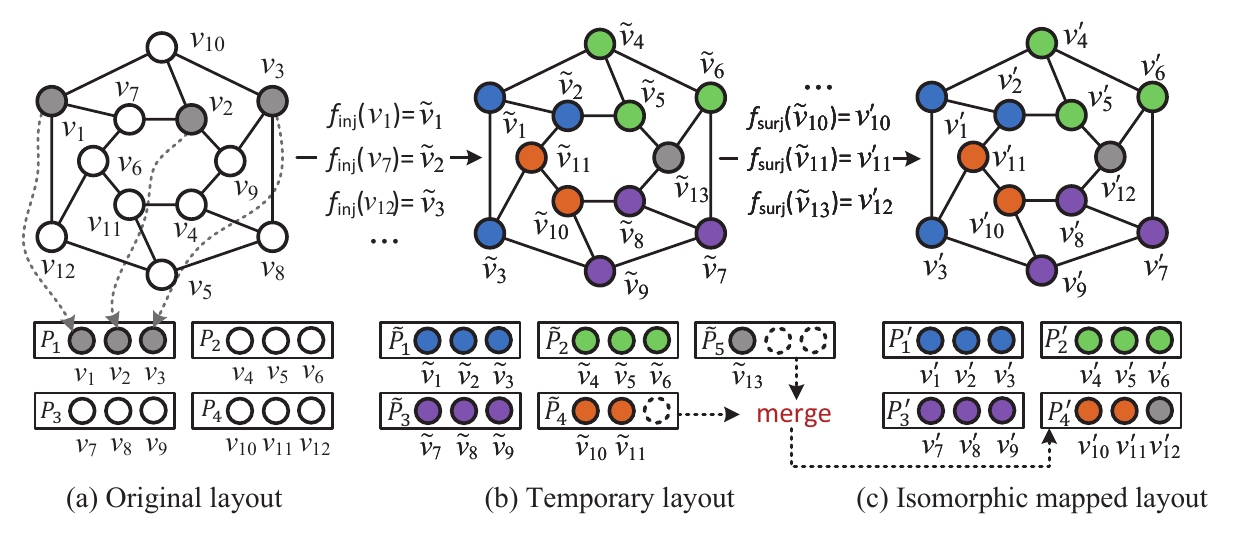}
    \vspace{-0.8cm}
    \caption{Isomorphic mapping on a graph index}
    \label{fig:order_step}
    \vspace{-0.1cm}
\end{figure}

\vspace{-0.1cm}
\begin{definition}[Isomorphic mapping]
\label{def:isomorphic}
Given two SSD layouts $L$ and $L'$ for graphs $G=(V,E)$ and $G'=(V',E')$. We say $f: L.{\rm IDs} \rightarrow L'.{\rm IDs}$ is an isomorphic mapping on their logical view (or logical layout), iff $f$ is a bijection satisfying three conditions: (1) $L$ and $L'$ contain the same number of data blocks. (2) $\forall b_{v_i}, b_{v_j}\in L$, if their identities $v_i\neq v_j$, then $f(v_i)\neq f(v_j)$. (3) $\forall b_{v_i}, b_{v_j}\in L$, if $\exists e(v_i,v_j)\in E$, then $e(v'_i,v'_j)\in E'$ where $v'_i=f(v_i)$ and $v'_j=f(v_j)$.
\end{definition}

\vspace{-0.1cm}
Figure \ref{fig:order_step}(c) shows an isomorphic mapped SSD layout of a graph. Each data block $b_{v}$ of $v\in V$ has a mapped $b_{v'}$ of $v'\in V'$ and all the data blocks on SSD pages retain the ascending order of vertex ID, thus we can efficiently access SSD pages using the same addressing mode as DiskANN. Since there exists various isomorphic mapping $f$, we need to implement one that can increase the data value of SSD pages. So, we present our \textit{pack-merge-based method} to return a logical layout $L'.{\sf IDs}$ with an isomorphic mapping $f:L.{\sf IDs}\rightarrow L'.{\sf IDs}$. It consists of two steps: packing (Algorithm \ref{algorithm:packing}) and merging (Algorithm \ref{algorithm:merging}). Packing aims to return a temporary logical layout $\tilde{L}.{\sf IDs}$ by an injection $f_{\rm inj}$ from $L.{\sf IDs}$ (Figure \ref{fig:order_step}(b)) and merging aims to return a final logical layout $L'.{\sf IDs}$ by a surjection $f_{\rm surj}$ from $\tilde{L}.{\sf IDs}$. The isomorphic mapping $f=f_{\rm surj}(f_{\rm inj}(\cdot))$ is a bijection from $L.{\sf IDs}$ to $L'.{\sf IDs}$. 


\setlength{\textfloatsep}{0cm}
\begin{algorithm}[t]
\setstretch{0.95}
\small
\caption{Packing($G$, $L.{\sf IDs}$, $b$)}
\label{algorithm:packing}
\LinesNumbered
\KwIn{$G$, $L.{\sf IDs}$, $b$}
\KwOut{temporary logical layout $\tilde{L}.{\sf IDs}$ with $f_{\rm inj}$}
\tcp{{\footnotesize Initialization: line 1}}
${\sf Visit} \leftarrow \emptyset$, $\tilde{L}.{\sf IDs} \leftarrow \emptyset$, ${\sf newID} \leftarrow 1$, $f_{\rm inj}\leftarrow \emptyset$\;
\tcp{{\footnotesize Star packing: lines 2-18}}
\For{$v \in V$ ${\rm \&\&}$ $v\notin {\sf Visit}$}{
        ${\sf Visit} \leftarrow {\sf Visit} \cup \{v\}$ \;
        $\tilde{P}.{\sf IDs} \leftarrow \{v\}$ \tcc*[r]{{\scriptsize a new temporary logical page}}
        sort $N(v)\subseteq V$ in ascending order of PQ distance\;
        \tcp{{\footnotesize update $\tilde{P}.{\sf IDs}$ with at most $b$ vertices}}
        \For{$v_i \in N(v)$ ${\rm in\ order}$}{
            \eIf{$|\tilde{P}.{\sf IDs}|<b$}{
                $\tilde{P}.{\sf IDs} \leftarrow \tilde{P}.{\sf IDs} \cup \{v_i\}$ \;
                ${\sf Visit} \leftarrow {\sf Visit} \cup \{v_i\}$ \;
            }{
                break \;
            }
        }
        \If{$|\tilde{P}.{\sf IDs}|<b$}{
                pad $\tilde{P}.{\sf IDs}$ with zero \tcc*[r]{{\scriptsize page alignment}}
        }
        $\tilde{L}.{\sf IDs} \leftarrow \tilde{L}.{\sf IDs} \cup \tilde{P}.{\sf IDs}$
}
\tcp{{\footnotesize Injection from $L.{\sf IDs}\rightarrow \tilde{L}.{\sf IDs}$: lines 19-26}}
\For{$\tilde{P}.{\sf IDs} \subset \tilde{L}.{\sf IDs}$}{
    \For{$v_i \in \tilde{P}.{\sf IDs}$}{
        $j={\rm newID}$++\;
        $f_{\rm inj}{\sf .put(}v_i,\tilde{v}_j{\sf )}$ \tcc*[r]{{\scriptsize update with $f_{\rm inj}(v_i)=\tilde{v}_j$}}     
        update $\tilde{P}.{\sf IDs}$ with $\tilde{v}_j$\;
    }
    ${\sf newID}=(\lceil\frac{{\sf newID}}{b}\rceil-1)\cdot b+1$\;
}
return $\tilde{L}.{\sf IDs}$ and $f_{\rm inj}$ \;
\end{algorithm}

\setlength{\textfloatsep}{0cm}
\begin{algorithm}[t]
\setstretch{0.95}
\small
\caption{Merging($\tilde{L}.{\sf IDs}$, $b$)}
\label{algorithm:merging}
\LinesNumbered
\KwIn{$\tilde{L}.{\sf IDs}$, $b$}
\KwOut{final logical layout $L'.{\sf IDs}$ with $f_{\sf surj}$}
\tcp{{\footnotesize Initialization: lines 1-2}}
$L'\leftarrow \emptyset$, ${\sf newID} \leftarrow 1$, $f_{\rm surj}\leftarrow \emptyset$\;
sort temporary logical pages of $\tilde{L}.{\sf IDs}$ in descending order of page size, i.e., the number of non-zero logical blocks\;
\For{$\forall \tilde{P}_i.{\sf IDs} \subset \tilde{L}.{\sf IDs}$ ${\rm in \ order}$}{
    \eIf{$|\tilde{P}_i.{\sf IDs}|==b$}{
     	$P'_i.{\sf IDs} \leftarrow \tilde{P}_i.{\sf IDs}$ \tcc*[r]{{\scriptsize copy a logical page}}
    		$\tilde{L}.{\sf IDs}\leftarrow \tilde{L}.{\sf IDs}\setminus \tilde{P}_i.{\sf IDs}$ \;
    }{
    		\tcp{{\footnotesize FFD-based merge: lines 8-13}}
   		 \For{$\forall \tilde{P}_j.{\sf IDs} \in \tilde{L}.{\sf IDs}$ ${\rm in \ order}$}{
            \If{$|\tilde{P}_i.{\sf IDs}|+|\tilde{P}_j.{\sf IDs}|\leq b$}{
            		$P'.{\sf IDs} \leftarrow \tilde{P}_i.{\sf IDs}+\tilde{P}_j.{\sf IDs}$ \tcc*[r]{{\scriptsize merge}}
				$\tilde{L}.{\sf IDs}\leftarrow \tilde{L}.{\sf IDs}\setminus \tilde{P}_i.{\sf IDs}\cup\tilde{P}_j.{\sf IDs}$ \;                
            }
        }
    }
   \tcp{{\footnotesize Surjection from $\tilde{L}.{\sf IDs}\rightarrow L'.{\sf IDs}$: lines 15-21}}
	\For{$\forall \tilde{v}_i\in P'_.{\sf IDs}$}{
			$j={\rm newID}$++\;
        		$f_{\rm surj}{\sf .put(}\tilde{v}_i,v'_j{\sf )}$ \tcc*[r]{{\scriptsize update with $f_{\rm inj}(\tilde{v}_i)=v'_j$}}     
        		update $P'.{\sf IDs}$ with $v'_j$\;
	}   
	${\sf newID}=(\lceil\frac{{\sf newID}}{b}\rceil-1)\cdot b+1$\;
	$L'.{\sf IDs}\leftarrow L'.{\sf IDs}\cup P'.{\sf IDs}$\;
}
return $L'.{\sf IDs}$ and $f_{\rm surj}$\;
\end{algorithm}

\vspace{0.1cm}
\noindent\textbf{Packing stage.} Given a graph index $G=(V,E)$ and a page capacity $b$, Algorithm \ref{algorithm:packing} returns $\tilde{L}.{\sf IDs}$ with $f_{\rm inj}$ by three steps. 


\vspace{0.1cm}
\noindent\underline{Initialization.} We initialize a set ${\sf Visit}=\emptyset$ to avoid repeated visits, a temporary logical layout $\tilde{L}.{\sf IDs}=\emptyset$, a vertex ID iterator ${\sf newID}$ from 1, and an empty map $f_{\rm inj}$ (line 1).

\vspace{0.1cm}
\noindent\underline{Star packing.}  For each unvisited vertex $v\in V$, we add it to a temporary logical page $\tilde{P}.{\sf IDs}$ and mark it as a visited vertex (lines 3-4). Then, we add $v$'s $(b-1)$ nearest neibhors (using PQ distance) from $G$ to the same logic page $\tilde{P}.{\sf IDs}$ (lines 5-13). If $N(v)<b$, we pad $\tilde{P}.{\sf IDs}$ with zeros (lines 14-16). We next add $\tilde{P}.{\sf IDs}$ to $\tilde{L}.{\sf IDs}$ and repeat above (line 17) until all vertices have been visited. Since $v$ and its $(b-1)$ nearest neighbors belong to the same page, the induced graph of them is a star-derived graph. In \S \ref{sec:connectivity}, we provide an effectiveness analysis using the properties of star-derived graph. 


\vspace{0.1cm}
\noindent\underline{Injective mapping.} Given a $\tilde{L}.{\sf IDs}$, we obtain the injection $f_{\rm inj}$ as follows. First, for each vertex $v_i\in \tilde{P}.{\sf IDs}\subset \tilde{L}.{\sf IDs}$, we update $f_{\rm inj}$ with an item $f(v_i)=\tilde{v}_j$, where $j={\sf newID}$ (lines 20-24). We next update ${\sf newID}$ (line 25) and repeat above until all logical pages in $\tilde{L}.{\sf IDs}$ have been visited. 


\vspace{-0.15cm}
\begin{example}
\label{exp:packing}
Figure \ref{fig:order_step}(b) shows a temporary logical layout obtained from Figure\ref{fig:order_step}(a). We use the same color to represent vertices in the same temporary page. Given the page capacity $b=3$, we first assign $v_1$ and its two nearest neighbors $v_7,v_{12}$ to the first page, pad non-full pages with zeros (e.g., the 4th and 5th pages), and repeat this until all vertices are processed. Next, we update the vertex IDs and update the injection $f_{\rm inj}$ with $f_{\rm inj}(v_1)=\tilde{v}_1$, $f_{\rm inj}(v_2)=\tilde{v}_5$, $f_{\rm inj}(v_3)=\tilde{v}_6$, etc.
\end{example}


\noindent\textbf{Merging stage.} The temporary logical layout has one problem: some pages are not full (as we pad zeros for nodes having neighbors $<b$), and it invalidates the original addressing mode. We present a merging stage with the goal of implementing a surjection $f_{\rm surj}:\tilde{L}.{\sf IDs}\rightarrow L'.{\sf IDs}$ to combine data blocks from non-full pages to form a full page. Given a $\tilde{L}.{\sf IDs}$ and a page capacity $b$, we do merging (Algorithm \ref{algorithm:merging}) as follows.

\vspace{0.1cm}
\noindent\underline{Initialization.} We initialize a final logical layout $L'.{\sf IDs}=\emptyset$, a vertex ID iterator ${\sf newID}$ from 1, and a map $f_{\rm surj}$ (line 1).

\vspace{0.1cm}
\noindent\underline{FFD-based merge.} We first sort temporary logical pages of $\tilde{L}.{\sf IDs}$ in descending order of page size. The logical page's size is the number of non-zero logical blocks in it, denoted by $|\tilde{P}.{\sf IDs}|$. We retain the temporary logical pages with size $=b$ in $L'.{\sf IDs}$ (lines: 4-6 and 21) and merge others having size $<b$ to form new logical pages (lines: 8-13 and 21). Specifically, we merge two temporary logical pages by bin packing based on First-Fit-Decreasing (FFD): we iteratively merge the largest non-full logical page with another smaller logical page to form a new full logical page until no more merge can be performed.

\vspace{0.1cm}
\noindent\underline{Surjection mapping.} Given a new logical page $P'.{\sf IDs}$ that is retained from $\tilde{P}.{\sf IDs}$ (line 5) or merged by two logical pages from $\tilde{L}.{\sf IDs}$ (line 10), we obtain the surjection $f_{\rm surj}$ as follows. For every vertex $\tilde{v}$ from $P'.{\sf IDs}$, we update $f_{\rm surj}$ with an item $f_{\rm surj}(\tilde{v}_i)=v'_j$, where $j={\sf newID}$ (lines 15-19). Next, we update ${\sf newID}$ (line 20) and add $P'.{\sf IDs}$ to $L'.{\sf IDs}$ (line 21). We repeat above until all logical pages from $\tilde{L}.{\sf IDs}$ have been processed and return the final $L'.{\sf IDs}$ with $f_{\rm surj}$ (line 23).

\begin{example}
\label{exp:merging}
Figure \ref{fig:order_step}(c) shows the final layout obtained from Figure \ref{fig:order_step}(b). We retain the first three full pages in the final logical layout. For the 4th page, it contains only two valid vertices so that we merge it with the 5th page. Then, we update the vertex IDs and update the surjection $f_{\rm surj}$ with $f_{\rm surj}(\tilde{v}_{10})=v'_{10}$, $f_{\rm surj}(\tilde{v}_{11})=v'_{12}$, and $f_{\rm surj}(\tilde{v}_{13})=v'_{12}$, etc. 
\end{example}

\vspace{0.1cm}
\noindent\textbf{Update the SSD layout using $f_{\rm inj}$ and $f_{\rm surj}$.} Given the original logical layout $L.{\sf IDs}$ and the output final logical layout $L'.{\sf IDs}$ with two mappings $f_{\sf inj}$ and $f_{\sf surj}$, we update SSD layout $L'$ with the real data blocks as follows. Given a logical page $P.{\sf IDs}\subset L.{\sf IDs}$, for each vertex $v_i\in P.{\sf IDs}$ with data block $b_{v_i}=\langle \vec{x}_{i},N(v_i)\rangle$, we first get $v_i$'s mapping vertex $v'_j=f_{\rm surj}(f_{\rm inj}(v_i))$. Then, we form $b_{v'_j}$ as $\langle \vec{x}_{v'_j},N(v'_j)\rangle$, where $\vec{x}_{v'_j}=\vec{x}_{v_i}$ because both $v_i$ and $v'_j$ represent the same vertex but with different IDs. For each vertex from $N(v_i)$, we obtain its mapping vertex using $f_{\rm surj}(f_{\rm inj}(\cdot))$ and add it to $N(v'_j)$. Finally, we write all reformed data blocks to their corresponding positions in $L'$ according to $L'.{\sf IDs}$. 

\subsection{Effectiveness Analysis of Refined Layout}
\label{sec:connectivity}
Since our intention is to increase the data value of each SSD page, it is necessary to evaluate the compactness of an SSD page after isomorphic mapping. We present a new metric called \textit{page compactness} based on two widely used metrics: \textit{diameter} and \textit{algebraic connectivity} \cite{Fiedler1973}. Graph diameter is defined as the longest shortest path between any two vertices of a graph. The larger the diameter, the less the closeness between any two vertices. Given an SSD page consisting of $b$ data blocks, we compute the diameter of the induced graph $G[V_b]$ of these $b$ vertices by Eq. \ref{eq:diameter}, where \textit{\text{shortest\_path}}$(u,v)$ returns the length of the shortest path between $u$ and $v$.
\begin{equation}
\label{eq:diameter}
\textit{\text{diam}}(G[V_b]) = \max_{u, v \in V_b} \textit{\text{shortest\_path}}(u, v)
\end{equation}

Algebraic connectivity reflects the global connectivity of a graph. It is the second-smallest eigenvalue of the Laplacian matrix \cite{Fiedler1973} of a graph. Given a induced graph $G[V_b]$ of an SSD page's $b$ vertices, it's algebraic connectivity $\lambda_2(G[V_b])$ is computed by Eq. \ref{eq:algebraic}, where $Lap(G[V_b])$ is the Laplacian matrix of $G[V_b]$ and $\xi$ is the eigenvector of $Lap(G[V_b])$.
\begin{equation}
\label{eq:algebraic}
\lambda_2(G[V_b])=\min_{\genfrac{}{}{0pt}{3}{\xi \perp \mathbf{1}}{\xi \neq \mathbf{0}}} \frac{\xi^T Lap(G[V_b])\ \xi}{\xi^T \xi}
\end{equation}

Given a graph $G[V_b]$, its Laplacian matrix is computed as
\begin{equation}
\label{eq:laplacian}
Lap_{(i,j)}(G[V_b])=
\begin{cases}
        deg(v_i) & \text{if } i=j \\
        -1 & \text{if } i\neq j \text{, } v_i \text{ is adjacent to } v_j\ . \\
        0  & \text{otherwise}
\end{cases}
\end{equation}

Given above two metrics, we define page compactness as

\begin{equation}
\label{eq:cohesiveness}
\gamma(G[V_b]) = \frac{\lambda_2(G[V_b])}{\textit{\text{diam}}(G[V_b])}\quad .
\end{equation}

Note that, the smaller (larger) the diameter (algebraic connectivity), the greater the page compactness. Table \ref{tab:compactness} provides the page compactness of the original SSD layout and isomorphic mapped SSD layout of the Vamana built on three datasets using the same $R=32$. In DiskANN, $R$ is the largest out-degree of Vamana used to control the index construction. Since DiskANN assigns data blocks on SSD in a round-robin, vertices in an SSD page are almost disconnected in Vamana. As a result, many SSD pages' algebraic connectivity is close to zero, resulting in a high probability that the page compactness tends to zero. For ours, we assign a vertex's $b-1$ nearest neighbors to the same page, so that the induced graph of them is a typical star-derived graph. Given this premise, we prove that the page compactness of ours must $>0.5$ in Theorem \ref{th:compactness}.

\begin{definition}
\label{def:star}
\textbf{Star-derived Graph}. Given a graph $G=(V,E)$ with one central vertex $v\in V$ and other $|V|-1$ peripheral vertices $V\setminus v$. We call $G$ a star graph if all peripheral vertices have edges to $v$ but no edges among themselves. Given another graph $G'=(V,E')$ with the same vertices as $G$ and $E'\supset E$, $G'$ is a star-derived graph, i.e., derived from a star graph $G$.
\end{definition}

\begin{theorem}
\label{th:compactness}
Given an SSD page of the isomorphic mapped layout of Vamana, its page compactness must $>0.5$.
\end{theorem}

\begin{proof}
\label{pf:compactness}
A star graph's diameter is fixed to 2 and its algebraic connectivity is constantly at 1. Since a star-derived graph allows connections among peripheral vertices, it may decrease the diameter (i.e.,  a diameter $\textit{\text{diam}}(\cdot)\leq 2$) and it would certainly increase the algebraic connectivity (i.e., $\lambda_2(\cdot)>1$). As a result, the page compactness must $>0.5$. 
\end{proof}

\begin{table}[t]
    \caption{Page compactness of two SSD layouts (original and isomorphic) for Vamana built on different datasets.}
    \label{tab:compactness}
    \begin{adjustbox}{width=\linewidth}
    \renewcommand{\arraystretch}{1.2}
    \begin{tabular}{c||c|c|c}
    \textbf{SSD layout $\downarrow$} & \textbf{sift100M (R32)} & \textbf{deep100M (R32)} & \textbf{turing100M (R32)} \\ \hline \hline
    original layout          & 0.000004       & 0              & 0                \\ \hline
    isomorphic mapped layout        & 0.658033       & 0.560141       & 0.547292         \\ 
    \end{tabular}
    \vspace{0.15cm}
    \end{adjustbox}
\end{table}

\vspace{-0.2cm}
\subsection{Complexity Analysis}
\label{sec:complexity_mapping}
\noindent\textbf{Complexity}. The complexity arises from both packing and merging. Assuming that basic arithmetic operations can be performed in $O(1)$ time, the complexity of packing is $O(|V|\cdot R\cdot d)$, where $V$ is the node set of a graph $G$, $R$ is the largest out-degree, and $d$ is the dimensions. This is because we need to calculate the PQ distances between each $v\in V$ and its $R$ neighbors and select the nearest $b-1$ neighbors to form an SSD page. The complexity of merging is evidently $O(b\cdot N)$. So, the overall time complexity of the mapping is $O(|V|\cdot R\cdot d + b\cdot N)$.

\begin{figure}
	\vspace{-0.2cm}
    \centering
    \includegraphics[scale=0.34]{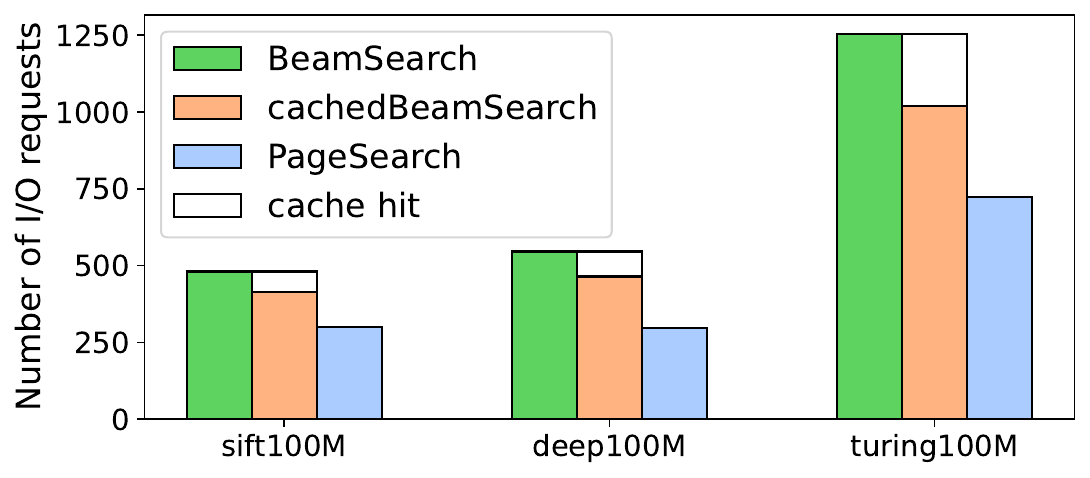}
    \vspace{-0.2cm}
    \caption{The number of I/O requests (SSD I/O and cache I/O) of beamsearch, cachedBeamsearch, and pagesearch.}
    \label{fig:cachehits}
\end{figure}

\section{Page-based search}
\label{sec:page_based_search}
We next discuss how to use the refined SSD layout of isomorphic mapped Vamana to accelerate the beamsearch. Straightforwardly, we can cache all the read SSD pages and check them before requesting SSD pages to avoid redundant SSD I/Os. It's worth mentioning that this would not reduce the total number of I/O requests but only replace a part of SSD I/O requests with cache I/O requests. The greater the cache hit rate, the larger the QPS is achieved. We implemented this method called cachedBeamsearch and compare it with beamsearch. Figure \ref{fig:cachehits} shows that cachedBeamsearch has the same number of I/O requests as beamsearch, of which only 10\%-20\% I/O requests are hit in cache and most of the cached SSD pages are unused for the node expansion of beamsearch. Moreover, the CPU remained largely underutilized during the search process due to the passive nature of cache requests. In order to take the advantages of refined SSD layout, \textit{we propose an active filtering-based asynchronous page expansion as a complement to node expansion, thus forming a new pagesearch.}

Figure \ref{fig:async} illustrates the pipelines of our pagesearch and beamsearch. Pagesearch relies on a meticulously customized page cache pool called page heap (\textbf{\S \ref{sec:pageheap}}). It involves four basic operators: \textit{Cache()}, \textit{Update()}, \textit{Check2ret()}, and \textit{Pop()} (the green components in Figure \ref{fig:async}(b)), based on which we design a page expansion strategy to actively filter more useful vertices as candidates for node expansion. Besides, we leverage the CPU's stall cycle (shown in Figure \ref{fig:async}(a)) to perform the proposed page expansion asynchronously, when the SSD I/O requests are processing at the same time. In this way, we improve the CPU utilization so as to improve overall QPS. Figure \ref{fig:cachehits} shows that pagesearch (right bar) achieves nearly 50\% reduction of SSD I/O requests compared with beamsearch.

\subsection{Page Heap}
\label{sec:pageheap}
PageHeap is a page cache pool with four basic operators.
\begin{itemize}
    \item \textbf{Cache()}. It caches a 4k-aligned page into the memory pool and register vertices of a page into a circular queue.  
    \item \textbf{Update()}. It first calculates the full vector distance between each vertex in a circular queue and a query. Then, it updates a min-heap with these vertices using the above distances and remove them from circular queue. This min-heap would be used in the page expansion (see  \S \ref{sec:pagesearch}.
    \item \textbf{Check2ret()}. It checks whether a given page’s data block exists in the memory pool. If the data block is found, it is returned as the output, otherwise it reports none.  
    \item \textbf{Pop()}. It popups the top-1 vertex having the minimum distance to the given query, from the min-heap.
\end{itemize}

\subsection{Pagesearch}
\label{sec:pagesearch}
\begin{figure}
	\vspace{-0.3cm}
    \centering
    \includegraphics[width=\linewidth]{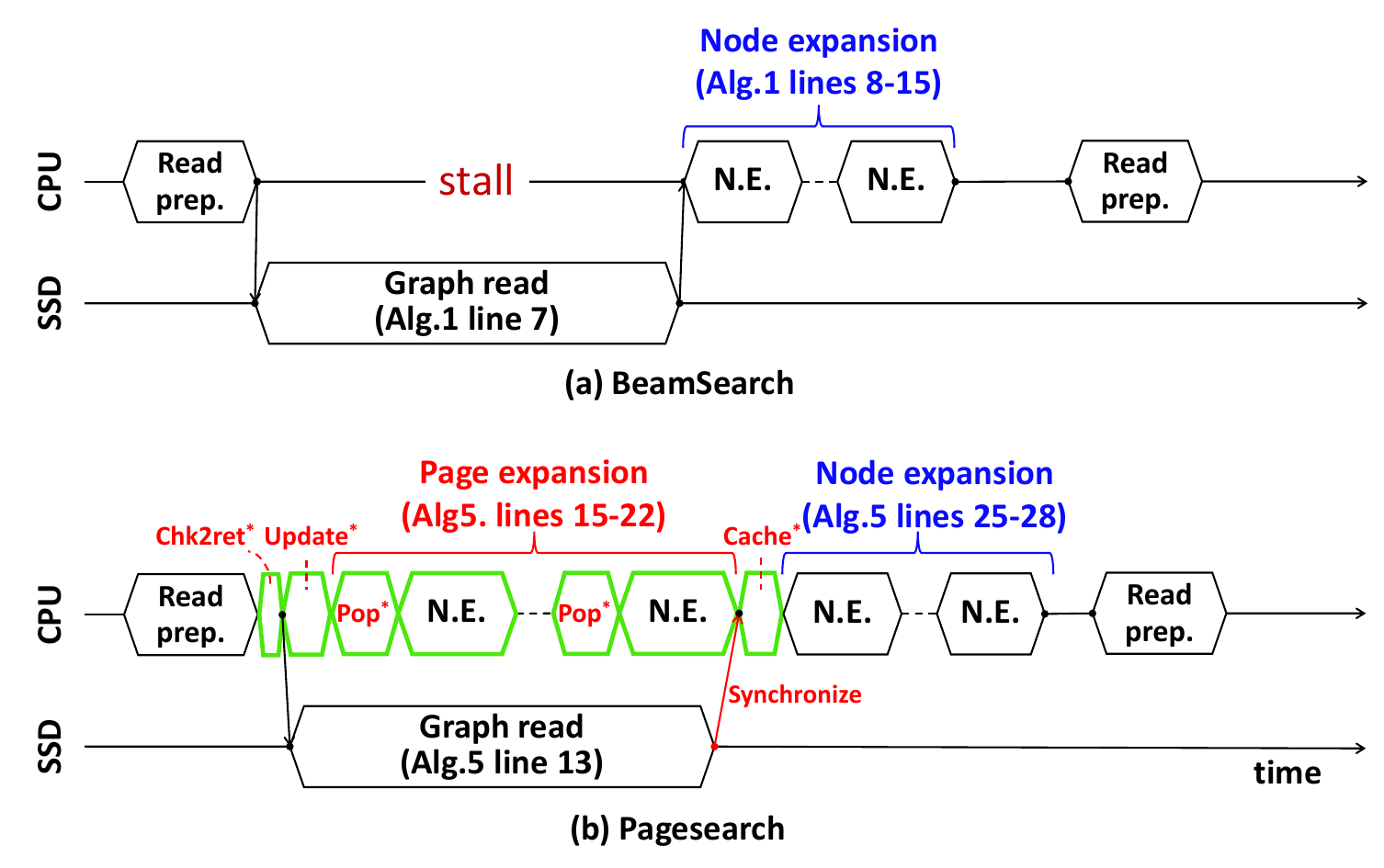}
    \vspace{-0.6cm}
    \caption{Pipeline of original beamsearch and our pagesearch.}
    \label{fig:async}
\end{figure}

\begin{algorithm}[t]
\setstretch{0.95}
\small
    \caption{Pagesearch($G$, $\vec{x}_{q}$, $v_{e}$, $B$, $L_s$, $k$)}
    \label{alg:pagesearch}
    \LinesNumbered
    \SetNoFillComment
    \KwIn{${G}$, $\vec{x}_{q}$, $v_{e}$, ${B}$, $L_s$, $k$ of top-k}
    \KwOut{approximate nearest top-$k$ neighbors $\mathcal{R}$ to $\vec{x}_q$}
    $\mathcal{C} \leftarrow \{v_e\}$ \tcc*[r]{{\scriptsize candidate set initialized as $\{v_e\}$}}
    $\mathcal{R} \leftarrow \emptyset$ \tcc*[r]{{\scriptsize top-$k$ results initialized as $\emptyset$}}
    $\mathcal{H} \leftarrow \emptyset$ \tcc*[r]{{\scriptsize PageHeap initialized as $\emptyset$}}
    \Do{$\mathcal{F}\neq\emptyset$}{
        $\mathcal{F} \leftarrow$ top-$B$ unvisited vertices from $C$ for expansion \;
        \tcp{\footnotesize prepare for SSD I/O requests: lines 6-12}
        $\mathcal{P} \leftarrow \emptyset$ \tcc*[r]{{\scriptsize page placeholders initialized as $\emptyset$}}
        \For{$v_i \in \mathcal{F}$}{
            \If{$\mathcal{H}\textnormal{.check2ret}(v_i)\textnormal{ is none}$}{
                $P_j \leftarrow$ register read for page containing $b_{v_i}$ \;
                $\mathcal{P} \leftarrow \mathcal{P} \cup P_j$ \;
            }
        }
        \tcp{\footnotesize async SSD I/O requests: line 13}
        async read all pages in $\mathcal{P}$ from SSD \;
        $\mathcal{H}$.update() \;
        \tcp{\footnotesize async page expansion: lines 15-22}
        \While{$b_{v_i}=\langle \vec{x}_{v_i},N(v_i) \rangle \leftarrow$ $\mathcal{H}$.{\rm pop()}\textnormal{ is not none}}{
            \If{$v_i$\textnormal{ is unvisited}}{
                NeighborExpansion($\vec{x}_{q}$, $b_{v_i},\mathcal{C},\mathcal{R},L_s,k$) \;
                \If{\textnormal{async read done}}{
                    break \;
                }
            }
        }
        wait for async read done \;
        $\mathcal{H}$.cache($\mathcal{P}$) \;
        \tcp{\footnotesize node expansion: lines 25-28}
        \For{\textnormal{unvisited }$v_i \in \mathcal{F}$}{
            $b_{v_i}=\langle \vec{x}_{v_i},N(v_i) \rangle \leftarrow$ $\mathcal{H}$.check2ret($v_i$) \;
            NeighborExpansion($\vec{x}_{q}$, $b_{v_i},\mathcal{C},\mathcal{R},L_s,k$) \;
        }
        return $\mathcal{R}$
    }
\end{algorithm}



Given a graph index $G=(V,E)$ of the vector dataset $\mathcal{X}$, a query vector $\vec{x}_q$, beam size $B$, search width $L_s$, and $k$ of top-$k$ results, Algorithm \ref{alg:pagesearch} shows the entire procedure. 

\vspace{0.1cm}
\noindent\underline{Initialization.} We initialize the candidate set $\mathcal{C}$ as $\{v_e\}$, the top-$k$ results $\mathcal{R}=\emptyset$, and the page heap $\mathcal{H}=\emptyset$ (lines 1-3). 

\vspace{0.1cm}
\noindent\underline{Read Preparing.} We prepare SSD I/O requests for unvisited top-$B$ vertices from $\mathcal{C}$, denoted by $\mathcal{F}$. (1) For each $v_i\in \mathcal{F}$, we invoke check2ret() to check if $v_i$ is cached and register a page placeholder $P_j$ into $\mathcal{P}$ (lines 6-12). (2) We submit the prepared read requests asynchronously (line 13). 

\vspace{0.1cm}
\noindent\underline{Page Expansion.} During asynchronous reading, we Update() page heap (line 14) and begin page expansion. We iteratively Pop() one vertex from page heap and invoke NeighborExpansion (Algorithm \ref{alg:neighbor_expand}) to update both $\mathcal{C}$ and $\mathcal{R}$ with more promising candidates (lines 15-22). We terminate the page expansion when asynchronous SSD read completes (lines 18-20), making it synchronized with the asynchronous read.

\vspace{0.1cm}
\noindent\underline{Node Expansion.} When page expansion finishes, we execute the same node expansion as beamsearch: taking unvisited node's data block $b_{v_i}$ from the memory pool of page heap by invoking check2ret(), for further NeighborExpansion. 

Pagesearch repeats the above read preparing, page expansion, and node expansion until no more vertices are visited (line 30) and returns the top-$k$ results $\mathcal{R}$ (line 29).

\section{Evaluation}
\label{sec:evaluation}
\begin{table}[t]
\vspace{-0.1cm}
\setlength{\abovecaptionskip}{0.1cm}
\setlength{\belowcaptionskip}{-0.3cm}
\centering
\small
    \caption{Statistics of datasets}
    \label{tab:datasets}
    \scalebox{0.85}{
    \begin{adjustbox}{width=\linewidth}
    \renewcommand{\arraystretch}{1.2}
    \begin{tabular}{c||c|c|c|c}
        \textbf{Dataset}               & \textbf{Dimension}           & \textbf{LID \cite{Facco2017}}                   & \textbf{Query}                & \textbf{Base (M:$10^6$, B:$10^9$)}    \\ \hline \hline
        image \footnote[1]                 & 100                 & 15.3                    & 10,000                  & 100M   \\ \hline
        sift \cite{texmex}                  & 128                 & 16.6                   & 10,000                  & 100M    \\ \hline
        deep \cite{Artem2016}				  &  96					& 17.6					& 10,000			& 10M$\sim$1B	\\ \hline
        msong \cite{Msong}                 & 420                 & 18.0                   & 200                  & 0.99M  \\ \hline
        crawl \cite{Crawl}                 & 300                 & 27.4                  & 10,000                  & 1.99M \\ \hline
        turing \cite{turing}                & 100                 & 30.5                    & 100,000                 & 100M   \\ \hline
        glove-100 \cite{Glove}             & 100                 & 34.3                  & 10,000                  & 1.18M \\ \hline
        gist \cite{texmex}                  & 960                 & 35.0                  & 1,000                   & 1M      
    \end{tabular}
    \end{adjustbox}
    }
\end{table}

We evaluate our DiskANN++, DiskANN, and other competitors on eight large-scale datasets, to answer five questions:

\noindent\textbf{Q1:} Does DiskANN++ achieve a better QPS vs. recall@$k$ than others within a low memory footprint (\textbf{\S \ref{sec:effect_efficiency}})?

\noindent\textbf{Q2:} Does DiskANN++ consistently outperform DiskANN under various hardware resource constraints (\textbf{\S \ref{sec:hardware}})?

\noindent\textbf{Q3:} How's the scalability of DiskANN++ w.r.t. different datasets (with various hardness) and data scales (\textbf{\S \ref{sec:scalability}})?

\noindent\textbf{Q4:} What's the parameter sensitivity of DiskANN++ (\textbf{\S \ref{sec:parameter}})?

\noindent\textbf{Q5:} What is the contribution of each component individually or in combination to DiskANN++ (\textbf{\S \ref{sec:ablation}}).


\subsection{Experiment Settings}
\noindent\textbf{Datasets.} Table \ref{tab:datasets} shows statistics of eight datasets, including seven public datasets (msong \cite{Msong}, glove-100 \cite{Glove}, crawl \cite{Crawl}, gist \cite{texmex}, sift \cite{texmex}, deep \cite{Artem2016}, and turing \cite{turing}) and one commercial dataset (image\footnote{\noindent Commercial image dataset provided by Huawei Technologies Co., Ltd.}). The slices of deep ranging from 1M to 1B were employed to evaluate the stability at scales.

\vspace{0.1cm}
\noindent\textbf{Comparing Algorithms.} We compared with three SSD-based solutions: (1) SPANN \cite{Chen2021}. (2) BBANN \cite{Simhadri2022}. (3) DiskANN \cite{JayaramSubramanya2019}. For ours, we implemented three versions: (4) DiskANN++ (w/o compression), (5) DiskANN++ (sq16), and (6) DiskANN++ (sq8). Here, sq16 and sq8 are compression ratio that means we compress a vector from fp32 to int16 or int8. The larger the number, the smaller the compression ratio. We study the effect of compression on DiskANN++ in \S \ref{sec:effect_efficiency}.

\begin{figure*}[t]
	\vspace{-0.5cm}
    \centering
    \includegraphics[scale=0.37]{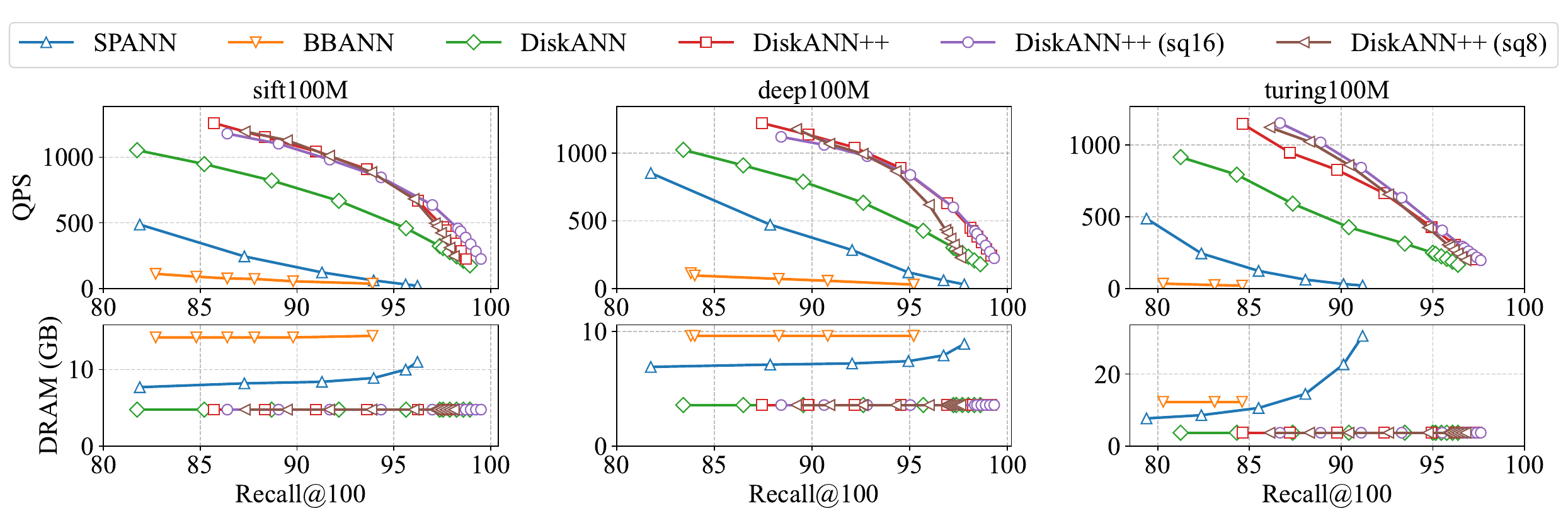}
    \vspace{-0.1cm}
    \caption{QPS and DRAM usage vs. Recall@100 for SPANN, BBANN, DiskANN, and DiskANN++ (with sq16/sq8 compression)}
    \label{fig:recall@100_vs_qps}
    \vspace{-0.4cm}
\end{figure*}

\begin{figure*}[t]
    \centering
    \includegraphics[scale=0.37]{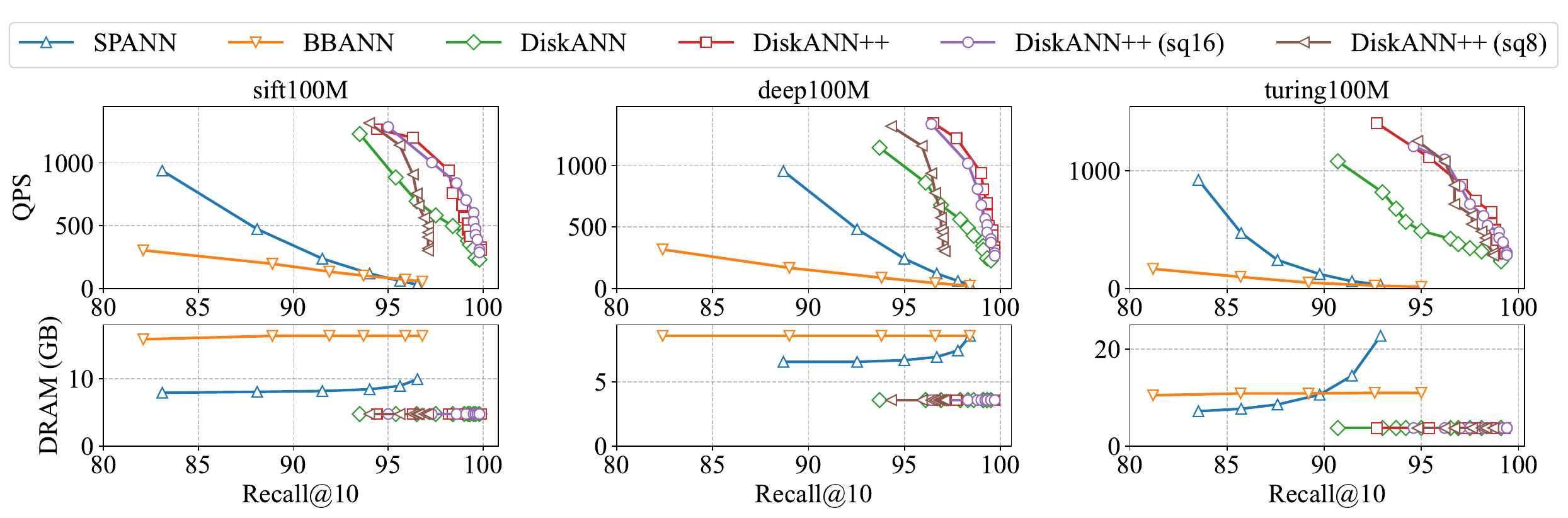}
    \vspace{-0.1cm}
    \caption{QPS and DRAM usage vs. Recall@10 for SPANN, BBANN, DiskANN, and DiskANN++ (with sq16/sq8 compression)}
    \label{fig:recall@10_vs_qps}
    \vspace{-0.4cm}
\end{figure*}

\begin{figure*}[t]
    \centering
    \includegraphics[scale=0.37]{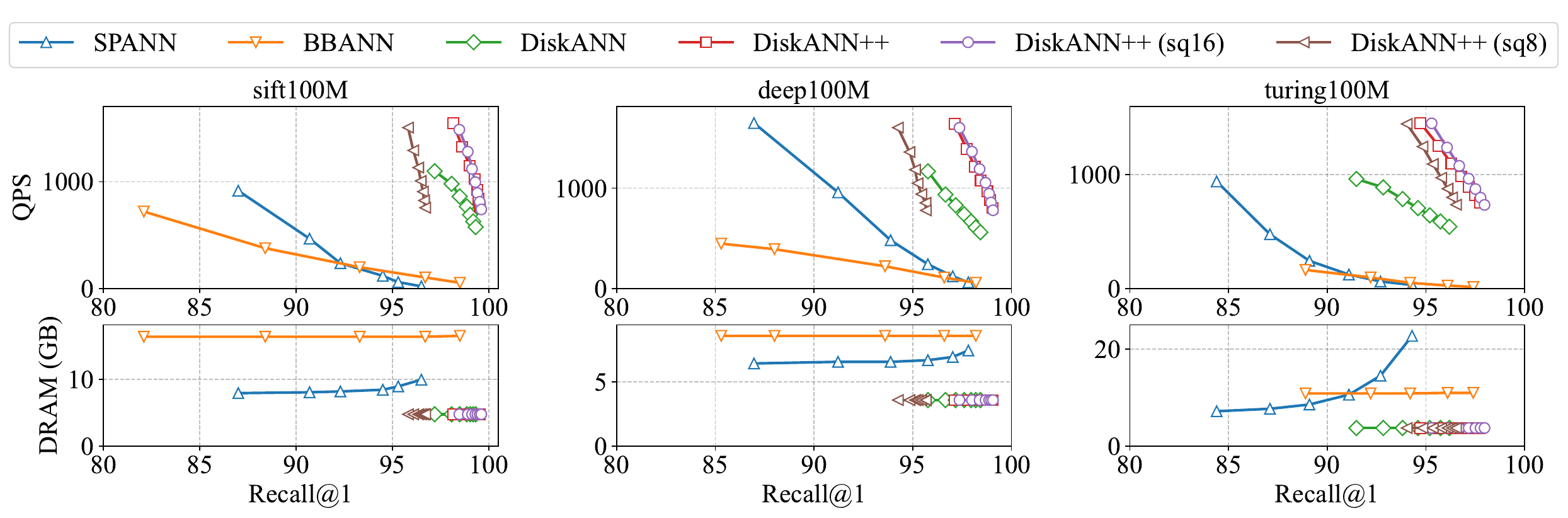}
    \vspace{-0.1cm}
    \caption{QPS and DRAM usage vs. Recall@1 for SPANN, BBANN, DiskANN, and DiskANN++ (with sq16/sq8 compression)}
    \label{fig:recall@1_vs_qps}
    \vspace{-0.4cm}
\end{figure*}

\vspace{0.1cm}
\noindent\textbf{Resource Constraints.} We default to employing 10\% of the dataset size as the memory constraint and use 8 threads for search (one thread per query). This configuration enables full usage of I/O resources when disk bandwidth and IOPS are not restricted. We also conducted experiments under varying hardware conditions (\S \ref{sec:hardware}) via Docker, including 1 to 8 threads (i.e., running 1 to 8 queries concurrently), 4k random I/O bandwidth (100-700MB/s), and 1-25\% of the dataset size as memory constraints.

\vspace{0.1cm}
\noindent\textbf{Evaluation Metrics.} We evaluate ANNS's effectiveness by \textit{recall@$k$} ($k=$ 1,10,100). We evaluate ANNS's efficiency by \textit{QPS} and \textit{speedup}. QPS exhibits an inverse correlation with each query's runtime. Speedup is defined as the ratio of QPS achieved by DiskANN++ and DiskANN at the same recall@$k$.

\vspace{0.1cm}
\noindent\textbf{Implementation Setup.} For SPANN, BBANN, and DiskANN, we use public implementations provided in their GitHub repositories \cite{SPANNrepo,BBANNrepo,DiskANNrepo}. All algorithms were compiled in C++17, retaining relevant SIMD and prefetching instructions. The benchmark scripts were implemented using Bash and Python3. All experiments were conducted on an Ubuntu 20.04 server with CPU-10C20T 3.70GHz, SSD-PM981 2TB, DRAM-32GB 2133Mb/s DDR4 x4.

\begin{figure*}
    \centering
    \includegraphics[width=0.95\linewidth]{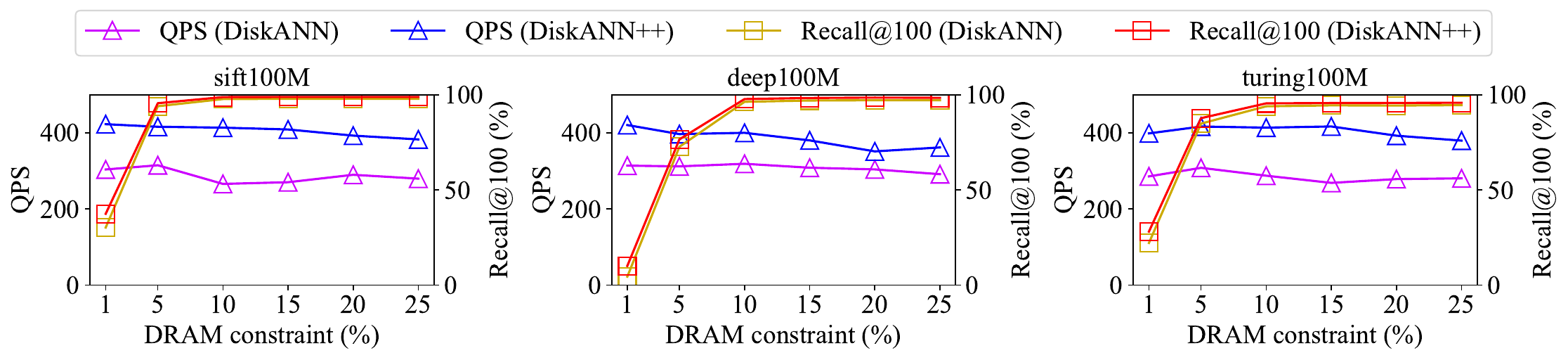}
    \vspace{-0.1cm}
    \caption{Effect of memory constraint on QPS and recall@100 of DiskANN and DiskANN++.}
    \label{fig:PQquality_res}
    \vspace{-0.6cm}
\end{figure*}

\noindent\textbf{Parameter Settings.} In both DiskANN and DiskANN++, we utilize identical graph construction parameters: R=32, L=500, $\alpha_1$=1.0, $\alpha_2$=1.2, and the same PQ construction method: 8-bit encoding (i.e., 256 pivots per chunk). The DRAM limit of index construction is set to 64GB, while in the index optimizing and searching stage, the DRAM limit are 10\% of the dataset size. Under such memory constraints, the construction of the index for 1 billion-sized slices is conducted in 10 slices and then merged, following the method outlined in the original DiskANN paper. In BBANN, we adopt index construction parameters based on the recommended configurations in its GitHub repository, with slight modifications: HNSW max layers is set to 32, candidate set list's size in construction is set to 512, and block size is 4096KB. In SPANN, we directly adopt the recommended parameter configuration it provides.

\subsection{Effectiveness and Efficiency Evaluation}
\label{sec:effect_efficiency}
We set the same memory constraint (10\% of the dataset size) for DiskANN and DiskANN++. For SPANN and BBANN, since they require at least 20\% of the dataset size in memory to gain an acceptable recall ($>$85\%) with recommended parameters, we set the memory constraint for them as 20\%.

\vspace{0.1cm}
\noindent\textbf{QPS vs. recall@$k$.} Figure \ref{fig:recall@100_vs_qps}-\ref{fig:recall@1_vs_qps} (top) demonstrate that DiskANN family of algorithms strikes a better balance between QPS and recall than that of SPANN and BBANN. DiskANN++ achieves a 50\%-100\% improvement trend under the same recall$@k$. For example, given the same recall@100 as 97\% in deep100M, the QPS of DiskANN++ and DiskANN are 605.3 and 310.2, leading a 95\% improvement on QPS. The same improvement hold for recall@1 and recall@10 as well. In compression scenarios, DiskANN++ (sq16) shows the similar trends as that without compression. Besides, the sq16 compression slightly improves the QPS given the same recall$@k$. This is because it reduces the word length of node data by compressing vector with less precision loss, which allows for accommodating more nodes in each SSD page. Consequently, it increases the page expansion width (i.e., $b$) in pagesearch to a certain extent, thus accelerating the convergence speed of  searching. On the other hand, we found that a noticeable drop occurs at high recall for sq8. This is because the significant precision loss of original vector when a strong compression is conducted.

\vspace{0.1cm}
\noindent\textbf{Memory Usage.} As shown in Figure \ref{fig:recall@1_vs_qps}-\ref{fig:recall@100_vs_qps} (bottom), SPANN and BBANN consume more memory than DiskANN family of algorithms and the memory usage increases as the precision increases. The reason is two-folded: (1) SPANN and BBANN inevitably require reading more disk clusters to achieve higher precision, resulting in ever-increasing memory consumption, and (2) SPANN and BBANN maintain SPT or graph index in memory, which inherently require more memory. In contrast, the DiskANN family of algorithms maintain a stable memory footprint that is lower than others because the PQ is used to obtain low-dimensional quantized vectors resident in memory.

\subsection{Effect of Hardware Resources}
\label{sec:hardware}

\begin{figure}[t]
	\vspace{-0.1cm}
    \centering
    \includegraphics[width=0.93\linewidth]{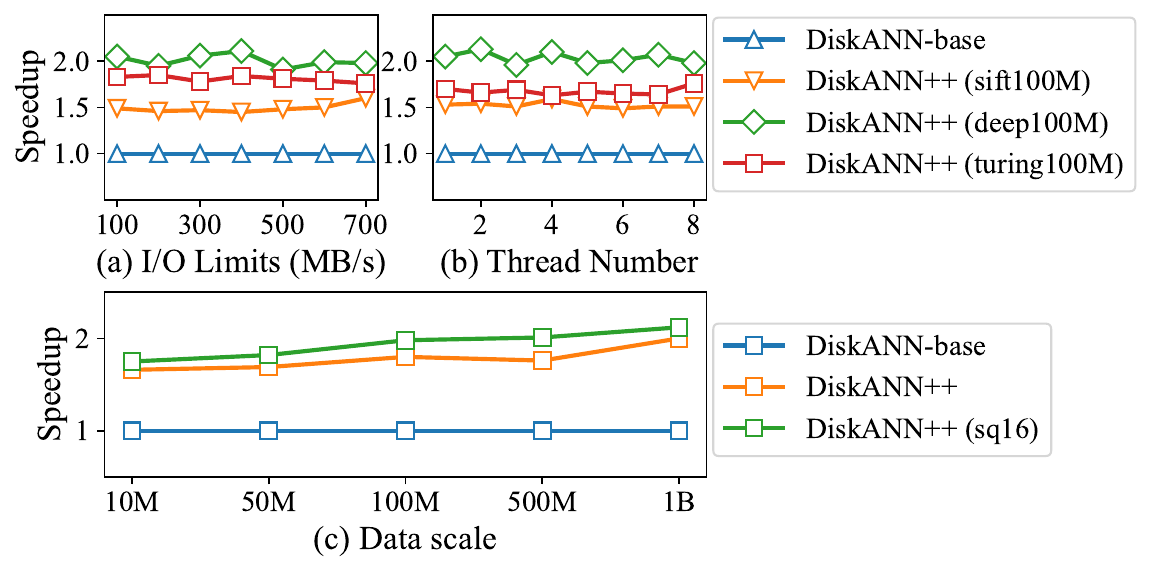}
    \vspace{-0.2cm}
    \caption{Speedup of DiskANN++ of different (a) I/O bandwidth, (b) \# thread, and (c) dataset scale (deep: 10M$\sim$1B).}
    \label{fig:stability_res}
\end{figure}

\noindent\textbf{Effect of I/O limits and \# threads.}  We first study the effect of hardware resources on DiskANN++ at a high recall@100 of 97\% on deep100M, sift100M, and turing100M, under varying I/O bandwidth (Figure \ref{fig:stability_res}(a)) and \# threads (Figure \ref{fig:stability_res}(b)). From Figure \ref{fig:stability_res}(a)-(b), we found that DiskANN++ achieves consistent speedup across varying I/O limits and \# threads (one thread per query) on the same dataset. The differences of speedup across datasets is attributed to the impact of different vector dimensions and local intrinsic dimensionality (LID) on search lengths, as discussed in \cite{Li2019}. 

\vspace{0.1cm}
\noindent\textbf{Effect of DRAM usage constraints.} We study the impact of memory constraints on the QPS and recall@$k$. We provide results with fixed search parameters under memory constraints at 1\%-25\% of the dataset size, e.g., the total size of sift100M is 48GB, deep100M is 36GB, and turing is 37GB. As shown in Figure \ref{fig:PQquality_res}, DiskANN++ consistently achieves higher recall than DiskANN under various memory constraints. Additionally, it demonstrates higher search speed (note the QPS curves). We analyze that, within the DiskANN framework, the primary impact of memory constraints lies in the compression ratio of the PQ in memory, reflecting the loss of precision in the original vectors. The larger the memory constraint, the less the PQ precision's loss. This directly influences the direction selection at each hop during search. Lower PQ quality can result in longer routing search paths. Worse still, it may lead to the search being unable to correctly reach the query's neighbors, thereby failing to meet the desired recall requirements. From the results shown in Figure \ref{fig:PQquality_res}, we can conclude that the recall increases as memory constraint increases.

\subsection{Scalability Evaluation}
\label{sec:scalability}

\begin{table}[t]
\vspace{-0.1cm}
\setlength{\abovecaptionskip}{0.1cm}
\setlength{\belowcaptionskip}{0.1cm}
\centering
\small
    \caption{QPS and speedup results of different datasets}
    \label{tab:cross_datasets}
    \scalebox{0.85}{
    \begin{adjustbox}{width=\linewidth}
    \renewcommand{\arraystretch}{1.2}
    \begin{tabular}{cc||c|cc|cc}
        \multicolumn{2}{c||}{\textbf{Datasets $\downarrow$}}                  & \textbf{DiskANN} & \multicolumn{2}{c|}{\textbf{DiskANN++}}         & \multicolumn{2}{c}{\textbf{DiskANN++ (sq16)}}   \\ \hline \hline
        \multicolumn{1}{c|}{dataset}    & LID  & QPS     & \multicolumn{1}{c|}{QPS}     & speedup & \multicolumn{1}{c|}{QPS}     & speedup \\ \hline
        \multicolumn{1}{c|}{image}      & 15.3   & 77.53   & \multicolumn{1}{c|}{117.07}  & 1.51x   & \multicolumn{1}{c|}{143.57}  & 1.85x  \\ \hline
        \multicolumn{1}{c|}{sift100M}   & 16.6  & 354.1   & \multicolumn{1}{c|}{523.32}  & 1.48x   & \multicolumn{1}{c|}{599.18}  & 1.69x   \\ \hline
        \multicolumn{1}{c|}{deep100M}   & 17.6 & 310.24  & \multicolumn{1}{c|}{605.31}  & 1.95x   & \multicolumn{1}{c|}{617.77}  & 1.99x   \\ \hline
        \multicolumn{1}{c|}{msong}      & 18.0  & 435.99  & \multicolumn{1}{c|}{653.35}  & 1.50x   & \multicolumn{1}{c|}{829.48}  & 1.90x   \\ \hline
        \multicolumn{1}{c|}{crawl}      & 27.4 & 738.2   & \multicolumn{1}{c|}{1209.03} & 1.64x   & \multicolumn{1}{c|}{1318.63} & 1.79x   \\ \hline
        \multicolumn{1}{c|}{turing100M} & 30.5   & 110.70  & \multicolumn{1}{c|}{270.65}  & 2.44x   & \multicolumn{1}{c|}{289.77}  & 2.61x   \\ \hline
        \multicolumn{1}{c|}{glove-100}  & 34.3 & 90.66   & \multicolumn{1}{c|}{245.49}  & 2.71x   & \multicolumn{1}{c|}{278.83}  & 3.08x   \\ \hline
        \multicolumn{1}{c|}{gist}       & 35.0 & 348.32  & \multicolumn{1}{c|}{347.7}   & 1.00x   & \multicolumn{1}{c|}{474.93}  & 1.36x   
    \end{tabular}
    \end{adjustbox}
    }
\end{table}

We evaluated the scalability of DiskANN++ across different data scales (Figure \ref{fig:stability_res}(c)) and types of datasets (Table \ref{tab:cross_datasets}) given recall@100 of 97\%. In Figure \ref{fig:stability_res}(c), the speedup is stable as the data scale increases. It is attributed to the adaptability of graph index to various data scales. Some classical studies \cite{Malkov2016,Fu2017} have indicated that the time complexity of graph search is logarithmic which grows slowly with the increase of data scale. From Table \ref{tab:cross_datasets}, ours achieves at least 1.5X speedup on most of the datasets except gist, due to its high vector dimensionality of 960, which results in only one vertex per page, limiting the effectiveness of pagesearch. However, after augmenting the node capacity to 2 with node compression of sq16, we achieved 1.36x improvement. We also noticed that our DiskANN++ is scalable to the hardness of datasets. For turing100M and glove-100 datasets that have a LID $>30$ (with a modest dimensionality of 100), we achieve at least 2.44x speedup. This is because when searching on the dataset with a larger LID, NN refine phase constitutes a larger proportion and our pagesearch can effectively reduce the I/O requests.




\subsection{Parameter Sensitivity}
\label{sec:parameter}

\begin{table}[t]
\centering
\small
    \caption{Effect of $N_{\rm cluster}$ on speedup under different I/O bandwidth (varied from 100 MB/s to 700 MB/s).}
    \label{tab:cross_ncluster}
    \scalebox{0.9}{
    \begin{adjustbox}{width=\linewidth}
    \renewcommand{\arraystretch}{1.2}
    \begin{tabular}{c||c|c|c|c|c|c|c}
        \multirow{2}{*}{$N_{\rm cluster}\downarrow$} & \multicolumn{7}{c}{\textbf{I/O bandwidth (MB/s)}}                                                                                     \\ 
    \cline{2-8}
                                                 & 100             & 200             & 300             & 400             & 500             & 600             & 700              \\ 
    \hline
    4096                                         & 1.176x          & 1.065x          & 1.085x          & 1.189x          & 1.127x          & 1.213x          & \textbf{1.221x}  \\ 
    \hline
    8192                                         & 1.156x          & 1.086x          & 1.126x          & 1.132x          & 1.175x          & 1.305x          & 1.166x           \\ 
    \hline
    16384                                        & 1.180x          & 1.113x          & 1.190x          & 1.237x          & 1.219x          & \textbf{1.334x} & 1.147x           \\ 
    \hline
    32768                                        & 1.189x          & 1.159x          & \textbf{1.225x} & \textbf{1.354x} & \textbf{1.418x} & 1.238x          & 1.077x           \\ 
    \hline
    65536                                        & \textbf{1.240x} & \textbf{1.161x} & 1.223x          & 1.275x          & 1.153x          & 1.124x          & 0.915x          
    \end{tabular}
    \end{adjustbox}
    }
\end{table}

\begin{figure}[t]
    \centering
    \includegraphics[width=0.9\linewidth]{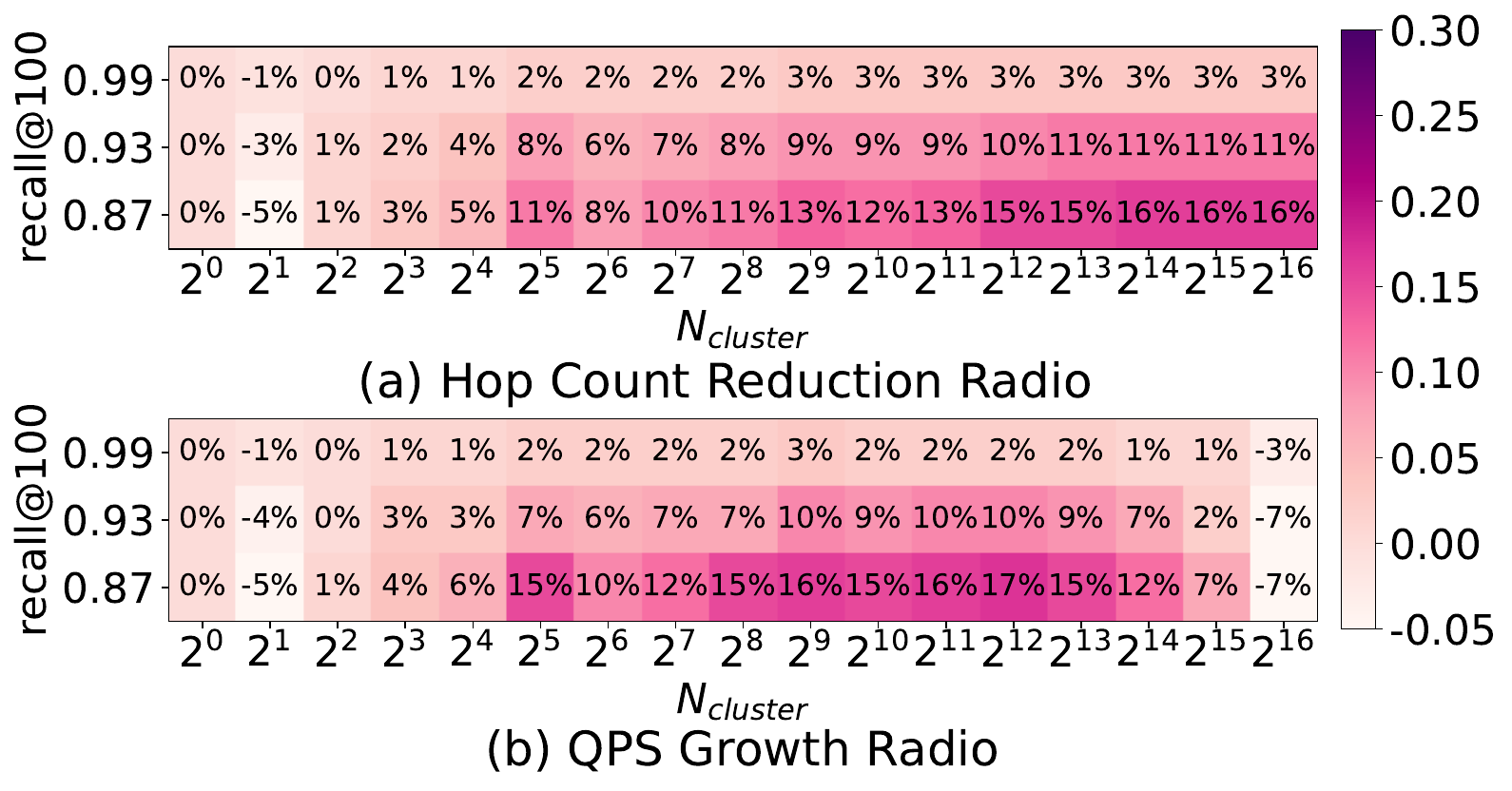}
    \vspace{-0.2cm}
    \caption{Effect of the query-sensitive entry vertex (with different $N_{\rm cluster}$) under different recall@100 (deep100M).}
    \label{fig:ncluster_vs_opt}
    \vspace{-0.1cm}
\end{figure}

\begin{figure}[t]
    \centering
    \includegraphics[width=0.92\linewidth]{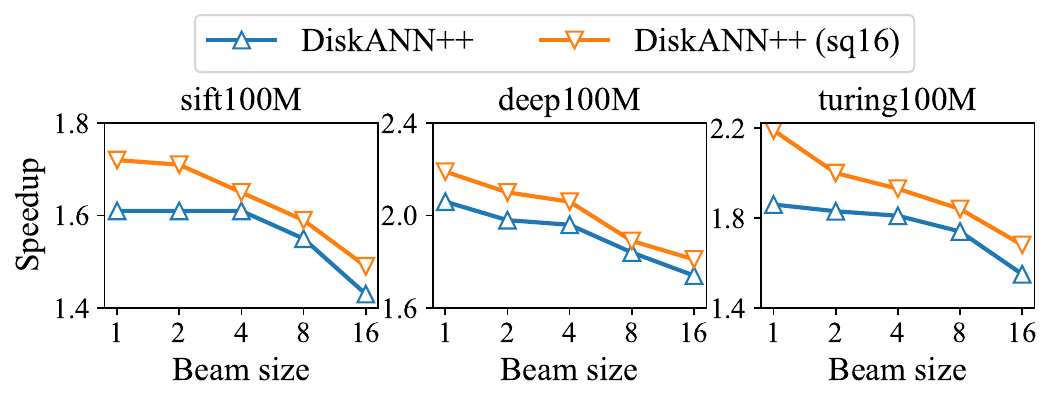}
    \vspace{-0.2cm}
    \caption{Effect of beamsize $B$ on speedup (deep100M).}
    \label{fig:beamtest}
\end{figure}

\noindent\textbf{Effect of $N_{\rm cluster}$.} The cluster size $N_{\rm cluster}$ has a significant impact on query-sensitive entry vertex selection strategy. So, we study its influence in Figure \ref{fig:ncluster_vs_opt}. We found that increasing $N_{\rm cluster}$ can effectively reduce routing hops and improve the QPS, which aligns with the conclusion in Theorem \ref{th:1} (\S \ref{sec:analysis_entry_vertex}). However, the search performance does not infinitely improve as $N_{\rm cluster}$ increases. This is because the time spent in entry vertex selection cannot be ignored. For example, using an excessively large $N_{\rm cluster}=2^{16}$ would increase the time on entry vertex selection, leading to a bad search performance. Therefore, it is necessary to carefully adjust $N_{\rm cluster}$ to achieve a favorable trade-off. Furthermore, as search precision (recall@100) increases, the QPS improvement brought by query-sensitive entry vertex selection strategy diminishes. This is because the routing procedure requires more time in NN refine step than NN approaching step to find the top-$k$ with a high recall, thus weakening the improvement brought by query-sensitive entry vertex selection. We also study the effect of $N_{\rm cluster}$ on speedup under different I/O bandwidth (Table \ref{tab:cross_ncluster}). It demonstrates that a smaller $N_{\rm cluster}$ is preferred to decrease the overhead of online entry vertex selection, when SSD's I/O bandwidth is large (or SSD I/O is fast). Otherwise, a larger $N_{\rm cluster}$ is better. Moreover, we tested the improvement in QPS under the maximum achievable bandwidth (with $N_{\rm cluster}$ = 4096). We found that query-sensitivity entry vertex strategy only consumes 0.04\% of the memory usage and 0.64\% of the runtime, but also offers a 1.12X speedup in QPS.


\vspace{0.1cm}
\noindent\textbf{Effect of beam size $B$.} In Figure \ref{fig:beamtest}, speedup decreases as the beam size increases. This is because increasing the beam size explicitly expands the width of node expansion, allowing node expansion to process more data blocks from page heap in each iteration, resulting in fewer unprocessed data blocks in page heap can be used for page expansion. However, as mentioned in the original DiskANN paper, continuously increasing the beam size is not conducive to load balancing for SSD across different queries (the original paper suggests $B=2,4,8$ for better performance). Therefore, under such parameter settings, DiskANN++ still exhibits at least 1.5X speedup.

\begin{table*}[t]
\vspace{-0.3cm}
    \centering
      \caption{Overhead and speedup (at recall@100=97\%) of randomOrder, parallelGorder, pack-merge (ours) using pagesearch.}
      \label{tab:overhead_of_reorder}
      \scalebox{0.85}{
      \begin{adjustbox}{width=\linewidth}
        \renewcommand{\arraystretch}{1.2}
        \begin{tabular}{c||ccc|ccc|ccc}
                                  & \multicolumn{3}{c|}{DRAM usage (GB)} & \multicolumn{3}{c|}{Runtime (sec)} & \multicolumn{3}{c}{Speedup with pagesearch}\\ \hline \hline
          graph index             & \multicolumn{1}{c|}{randomOrder}     & \multicolumn{1}{c|}{parallelGorder}       & pack-merge                                   & \multicolumn{1}{c|}{randomOrder} & \multicolumn{1}{c|}{parallelGorder} & pack-merge & \multicolumn{1}{c|}{randomOrder} & \multicolumn{1}{c|}{parallelGorder} & pack-merge \\ \hline
          Vamana (sift100M R32)   & \multicolumn{1}{c|}{2.2}             & \multicolumn{1}{c|}{75.1 (MLE)}           & 6.3                                           & \multicolumn{1}{c|}{449.5}       & \multicolumn{1}{c|}{2739.1}         & 222.3       & \multicolumn{1}{c|}{1.02x}       & \multicolumn{1}{c|}{1.84x}          & 1.69x       \\ \hline
          Vamana (deep100M R32)   & \multicolumn{1}{c|}{1.9}             & \multicolumn{1}{c|}{72.5 (MLE)}           & 5.1                                           & \multicolumn{1}{c|}{381.2}       & \multicolumn{1}{c|}{2651.4}         & 123.1       & \multicolumn{1}{c|}{1.08x}       & \multicolumn{1}{c|}{2.15x}          & 2.02x       \\ \hline
          Vamana (turing100M R32) & \multicolumn{1}{c|}{2.0}             & \multicolumn{1}{c|}{73.4 (MLE)}           & 6.1                                           & \multicolumn{1}{c|}{366.7}       & \multicolumn{1}{c|}{2644.8}         & 195.2       & \multicolumn{1}{c|}{1.08x}       & \multicolumn{1}{c|}{2.08x}          & 1.88x
        \end{tabular}
      \end{adjustbox}
      }
      \vspace{-0.5cm}
    \end{table*}

\begin{table}[t]
    \small
    \setstretch{1}
    \caption{Ablation analysis on query-sensitive entry vertex (A), isomorphic mapping (B), and pagesearch (C).}
    \label{tab:ablation_res}
    \begin{adjustbox}{width=\linewidth}
    \renewcommand{\arraystretch}{1.2}
    \begin{tabular}{ccccccc}
    \multicolumn{1}{c||}{Dataset}    & \multicolumn{3}{c|}{Recall@100=85\%}                                                                            & \multicolumn{3}{c}{Recall@100=97\%}                                                           \\ \hline
    \multicolumn{1}{c||}{sift100M}   & \multicolumn{3}{c|}{DiskANN++}                                                                                & \multicolumn{3}{c}{DiskANN++}                                                               \\ \hline \hline
    \multicolumn{1}{c||}{components} & \multicolumn{1}{c|}{QPS}             & \multicolumn{1}{c|}{mean I/Os} & \multicolumn{1}{c|}{mean hops} & \multicolumn{1}{c|}{QPS}            & \multicolumn{1}{c|}{mean I/Os} & mean hops     \\ \hline
    \multicolumn{1}{c||}{-}          & \multicolumn{1}{c|}{946.42 (1.00x)}  & \multicolumn{1}{c|}{181.54}    & \multicolumn{1}{c|}{24.66}            & \multicolumn{1}{c|}{385.88 (1.00x)} & \multicolumn{1}{c|}{447.49}    & 57.61                \\ \hline
    \multicolumn{1}{c||}{A}          & \multicolumn{1}{c|}{1033.64 (1.09x)} & \multicolumn{1}{c|}{164.57}    & \multicolumn{1}{c|}{22.51}            & \multicolumn{1}{c|}{401.14 (1.04x)} & \multicolumn{1}{c|}{431.58}    & 55.62                \\ \hline
    \multicolumn{1}{c||}{B}          & \multicolumn{1}{c|}{959.8 (1.01x)}   & \multicolumn{1}{c|}{181.54}    & \multicolumn{1}{c|}{24.66}            & \multicolumn{1}{c|}{387.17 (1.00x)} & \multicolumn{1}{c|}{447.49}    & 57.61                \\ \hline
    \multicolumn{1}{c||}{C}          & \multicolumn{1}{c|}{1065.09 (1.13x)} & \multicolumn{1}{c|}{179.24}    & \multicolumn{1}{c|}{24.38}            & \multicolumn{1}{c|}{440.34 (1.14x)} & \multicolumn{1}{c|}{438.87}    & 56.6                 \\ \hline
    \multicolumn{1}{c||}{AB}         & \multicolumn{1}{c|}{1022.24 (1.08x)} & \multicolumn{1}{c|}{164.57}    & \multicolumn{1}{c|}{22.51}            & \multicolumn{1}{c|}{397.41 (1.03x)} & \multicolumn{1}{c|}{431.58}    & 55.62                \\ \hline
    \multicolumn{1}{c||}{AC}         & \multicolumn{1}{c|}{1152.38 (1.22x)} & \multicolumn{1}{c|}{163.79}    & \multicolumn{1}{c|}{22.41}            & \multicolumn{1}{c|}{447.22 (1.16x)} & \multicolumn{1}{c|}{425.35}    & 54.92                \\ \hline
    \multicolumn{1}{c||}{BC}         & \multicolumn{1}{c|}{1281.32 (1.35x)} & \multicolumn{1}{c|}{148.32}    & \multicolumn{1}{c|}{20.58}            & \multicolumn{1}{c|}{660.6 (1.71x)}  & \multicolumn{1}{c|}{291.85}    & 38.25                \\ \hline
    \multicolumn{1}{c||}{ABC}        & \multicolumn{1}{c|}{1412.35 (1.49x)} & \multicolumn{1}{c|}{133.41}    & \multicolumn{1}{c|}{18.72}            & \multicolumn{1}{c|}{691.88 (1.79x)} & \multicolumn{1}{c|}{278.65}    & 36.62                \\
    \multicolumn{1}{l}{}            & \multicolumn{1}{l}{}                 & \multicolumn{1}{l}{}           & \multicolumn{1}{l}{}                  & \multicolumn{1}{l}{}                & \multicolumn{1}{l}{}           & \multicolumn{1}{l}{} \\
    \multicolumn{1}{c||}{Dataset}    & \multicolumn{3}{c|}{Recall@100=85\%}                                                                            & \multicolumn{3}{c}{Recall@100=97\%}                                                           \\ \hline
    \multicolumn{1}{c||}{deep100M}   & \multicolumn{3}{c|}{DiskANN++}                                                                                & \multicolumn{3}{c}{DiskANN++}                                                               \\ \hline \hline
    \multicolumn{1}{c||}{components} & \multicolumn{1}{c|}{QPS}             & \multicolumn{1}{c|}{mean I/Os} & \multicolumn{1}{c|}{mean hops} & \multicolumn{1}{c|}{QPS}            & \multicolumn{1}{c|}{mean I/Os} & mean hops     \\ \hline
    \multicolumn{1}{c||}{-}          & \multicolumn{1}{c|}{945.69 (1.00x)}  & \multicolumn{1}{c|}{184.15}    & \multicolumn{1}{c|}{25.06}            & \multicolumn{1}{c|}{308.78 (1.00x)} & \multicolumn{1}{c|}{557.93}    & 71.42                \\ \hline
    \multicolumn{1}{c||}{A}          & \multicolumn{1}{c|}{1104.46 (1.17x)} & \multicolumn{1}{c|}{153.7}     & \multicolumn{1}{c|}{21.25}            & \multicolumn{1}{c|}{325.66 (1.05x)} & \multicolumn{1}{c|}{529.26}    & 67.91                \\ \hline
    \multicolumn{1}{c||}{B}          & \multicolumn{1}{c|}{928.85 (0.98x)}  & \multicolumn{1}{c|}{184.15}    & \multicolumn{1}{c|}{25.06}            & \multicolumn{1}{c|}{305.81 (0.99x)} & \multicolumn{1}{c|}{557.93}    & 71.42                \\ \hline
    \multicolumn{1}{c||}{C}          & \multicolumn{1}{c|}{1043.07 (1.10x)} & \multicolumn{1}{c|}{177.15}    & \multicolumn{1}{c|}{24.19}            & \multicolumn{1}{c|}{332.42 (1.08x)} & \multicolumn{1}{c|}{549.1}     & 70.37                \\ \hline
    \multicolumn{1}{c||}{AB}         & \multicolumn{1}{c|}{1096.92 (1.16x)} & \multicolumn{1}{c|}{153.7}     & \multicolumn{1}{c|}{21.25}            & \multicolumn{1}{c|}{324.25 (1.05x)} & \multicolumn{1}{c|}{529.26}    & 67.91                \\ \hline
    \multicolumn{1}{c||}{AC}         & \multicolumn{1}{c|}{1156.4 (1.22x)}  & \multicolumn{1}{c|}{153.66}    & \multicolumn{1}{c|}{21.23}            & \multicolumn{1}{c|}{340.15 (1.10x)} & \multicolumn{1}{c|}{528.33}    & 67.74                \\ \hline
    \multicolumn{1}{c||}{BC}         & \multicolumn{1}{c|}{1296.92 (1.37x)} & \multicolumn{1}{c|}{142.83}    & \multicolumn{1}{c|}{19.92}            & \multicolumn{1}{c|}{588.93 (1.91x)} & \multicolumn{1}{c|}{316.78}    & 41.42                \\ \hline
    \multicolumn{1}{c||}{ABC}        & \multicolumn{1}{c|}{1523.41 (1.61x)} & \multicolumn{1}{c|}{118.89}    & \multicolumn{1}{c|}{16.99}            & \multicolumn{1}{c|}{625.23 (2.02x)} & \multicolumn{1}{c|}{291.78}    & 38.28                \\
    \multicolumn{1}{l}{}            & \multicolumn{1}{l}{}                 & \multicolumn{1}{l}{}           & \multicolumn{1}{l}{}                  & \multicolumn{1}{l}{}                & \multicolumn{1}{l}{}           & \multicolumn{1}{l}{} \\
    \multicolumn{1}{c||}{Dataset}    & \multicolumn{3}{c|}{Recall@100=85\%}                                                                            & \multicolumn{3}{c}{Recall@100=97\%}                                                           \\ \hline
    \multicolumn{1}{c||}{turing100M} & \multicolumn{3}{c|}{DiskANN++}                                                                                & \multicolumn{3}{c}{DiskANN++}                                                               \\ \hline \hline
    \multicolumn{1}{c||}{components} & \multicolumn{1}{c|}{QPS}             & \multicolumn{1}{c|}{mean I/Os} & \multicolumn{1}{c|}{mean hops} & \multicolumn{1}{c|}{QPS}            & \multicolumn{1}{c|}{mean I/Os} & mean hops     \\ \hline
    \multicolumn{1}{c||}{-}          & \multicolumn{1}{c|}{866.13 (1.00x)}  & \multicolumn{1}{c|}{202.21}    & \multicolumn{1}{c|}{27.22}            & \multicolumn{1}{c|}{136.12 (1.00x)} & \multicolumn{1}{c|}{1244.34}   & 157.5                \\ \hline
    \multicolumn{1}{c||}{A}          & \multicolumn{1}{c|}{990.49 (1.14x)}  & \multicolumn{1}{c|}{174.84}    & \multicolumn{1}{c|}{23.79}            & \multicolumn{1}{c|}{140.21 (1.03x)} & \multicolumn{1}{c|}{1215.26}   & 153.81               \\ \hline
    \multicolumn{1}{c||}{B}          & \multicolumn{1}{c|}{855.86 (0.99x)}  & \multicolumn{1}{c|}{202.21}    & \multicolumn{1}{c|}{27.22}            & \multicolumn{1}{c|}{136.33 (1.00x)} & \multicolumn{1}{c|}{1244.34}   & 157.5                \\ \hline
    \multicolumn{1}{c||}{C}          & \multicolumn{1}{c|}{988.52 (1.14x)}  & \multicolumn{1}{c|}{189.56}    & \multicolumn{1}{c|}{25.57}            & \multicolumn{1}{c|}{153.48 (1.13x)} & \multicolumn{1}{c|}{1203.58}   & 152.33               \\ \hline
    \multicolumn{1}{c||}{AB}         & \multicolumn{1}{c|}{958.94 (1.11x)}  & \multicolumn{1}{c|}{174.84}    & \multicolumn{1}{c|}{23.79}            & \multicolumn{1}{c|}{139.33 (1.02x)} & \multicolumn{1}{c|}{1215.26}   & 153.81               \\ \hline
    \multicolumn{1}{c||}{AC}         & \multicolumn{1}{c|}{1066.85 (1.23x)} & \multicolumn{1}{c|}{173.12}    & \multicolumn{1}{c|}{23.58}            & \multicolumn{1}{c|}{155.43 (1.14x)} & \multicolumn{1}{c|}{1187.15}   & 150.33               \\ \hline
    \multicolumn{1}{c||}{BC}         & \multicolumn{1}{c|}{1228.52 (1.42x)} & \multicolumn{1}{c|}{152.59}    & \multicolumn{1}{c|}{21.07}            & \multicolumn{1}{c|}{259.17 (1.90x)} & \multicolumn{1}{c|}{714.63}    & 91.37                \\ \hline
    \multicolumn{1}{c||}{ABC}        & \multicolumn{1}{c|}{1363.68 (1.57x)} & \multicolumn{1}{c|}{131.89}    & \multicolumn{1}{c|}{18.54}            & \multicolumn{1}{c|}{265.95 (1.95x)} & \multicolumn{1}{c|}{694.09}    & 88.84               
    \end{tabular}
    \end{adjustbox}
\end{table}

\subsection{Ablation Analysis}
\label{sec:ablation}
\noindent\textbf{Effect of each component.} We employ ablation analysis to assess the contributions of query-sensitive entry vertex (A), isomorphic mapping (B), and pagesearch (C) individually within DiskANN++. The experiments compare DiskANN++ and DiskANN across QPS, mean \# I/Os, and mean routing length (i.e., hops), given recall@100=97\% (high recall) and recall@100=85\% (low recall), across three common datasets: sift100M, deep100M, and turing100M. From Table \ref{tab:ablation_res}, we conclude that: 
(1) The query-sensitive entry vertex has a modest contribution at high recall and a more substantial impact at low recall. And its contribution is independent of other components. For enhanced persuasiveness, we record the hop reduction attributed to the query-sensitive entry vertex for each query in the official query set of deep \cite{Artem2016} (refer to Figure \ref{fig:hop_reduce_count}). The left subfigure illustrates the relationship between the hop reduction and the distance between the query and central vertex (medoid) of graph index, and the right subfigure provides a statistical overview of the query count for different hop reduction. It is evident that queries positioned farther from the medoid experience more substantial benefits. Practically, there is a higher probability of a random high-dimensional point falling at a greater distance compared to falling near a fixed point. Hence, on average, queries exhibit a pronounced reduction in hop count.
(2) Using isomorphic mapping alone does not accelerate beamsearch, because beamsearch doesn't have a mechanism to consider the highly-related vertices in the same page. Remark: despite the potential increase in CPU cache prefetching effects (refer to \cite{Wei2016}), in scenarios related to disks in this context, SRAM prefetch optimization is diluted to insignificance by a larger and slower proportion of I/O operations.
(3) Using pagesearch alone shows a limited improvement. This is because the original SSD layout is established by a round-robin assignment, resulting a low locality of vertices in the same page. So, pagesearch cannot benefit more from original SSD layout.
(4) The combination of isomorphic mapping and pagesearch represents the primary contribution, it brings more improvement on QPS at higher recall. This is because pagesearch can effectively unleash the potential of the refined SSD layout via isomorphic mapping.

\begin{figure}[t]
    \centering
    \includegraphics[width=\linewidth]{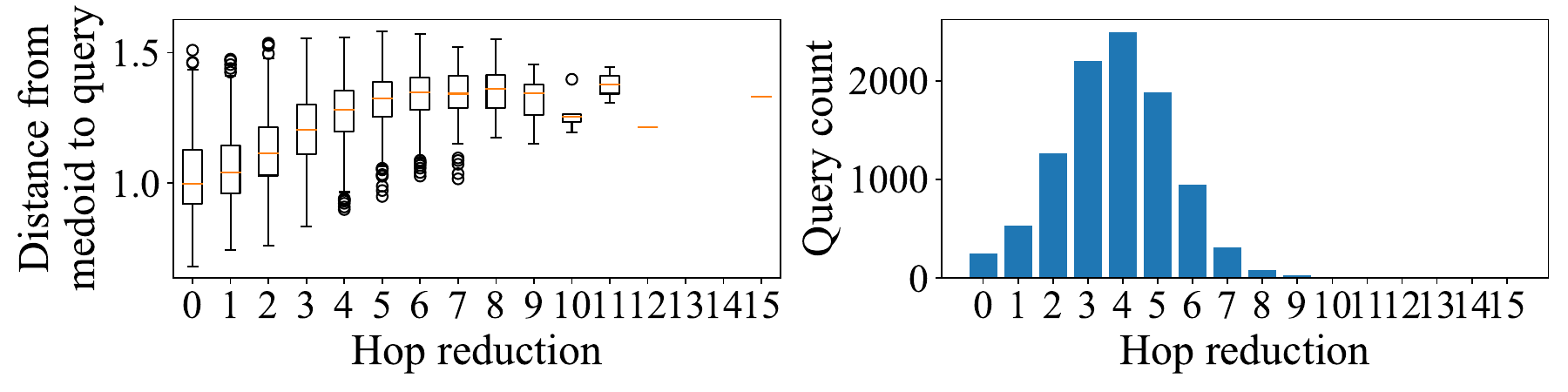}
    \vspace{-0.7cm}
    \caption{Relationship between the hop reduction and the distance between the query and graph medoid (left). Query count for different hop reduction (right).}
    \label{fig:hop_reduce_count}
\end{figure}

\vspace{0.1cm}
\noindent\textbf{Effect of different graph reordering methods.} The key of our isomorphic mapping is the pack-merge method (\S \ref{sec:pack}). We replace it by different graph reordering methods that can also be used to refine the SSD layout, e.g., \textit{randomOrder} and \textit{Gorder} \cite{Wei2016}, and evaluate their overhead (memory usage and runtime) and speedup, using the same pagesearch. Considering the immense time complexity of Gorder, rendering it impractical to complete in a foreseeable time, we adopted the improved parallelGorder method used in OOD-DiskANN \cite{Jaiswal2022}. This method sacrifices some reorder effect to enhance runtime speed. All requirements are conducted with 8 threads. 

From Table \ref{tab:overhead_of_reorder}, we found that our pack-merge method is the most efficient one that consumes modest memory and offers a comparable QPS improvement as parallelGorder+pagesearch. In contrast, since parallelGorder requires $>$ 70GB memory to perform, which largely exceeds the memory constraint (MLE stands for memory limit exceed), its practical deployability is significantly limited. Therefore, due to the lightweight nature of the pack-merge-based method, we can easily optimize the existing Vamana index on any device where DiskANN is deployed, whether it is the machine used for index construction or the one used for searching. This optimization can lead to a substantial improvement in search performance.

\section{Related Work}
We review the relevant technologies, including ANNS, graph-based ANNS, and DiskANN and its variants.

\noindent\textbf{Approximate nearest neighbor search (ANNS).} The researches of ANNS are generally categorized as three types: quantization-based \cite{Jegou2011,Ge2013,Kalantidis2014,Wang2017,Pan2020}, hash-based \cite{Gionis1999,Weiss2008,Xu2011,Huang2017,Karaman2019,Gong2020}, and graph-based \cite{Hajebi2011,Malkov2014,Fu2016,Malkov2016,Fu2017} ANNS. Among them, graph-based methods achieve a favourable tradeoff between accuracy and efficiency \cite{Hacid2010,Aoyama2013,Aumueller2020,Wang2021}. However, existing research has primarily focused on optimizing search accuracy and efficiency but ignore the significant memory overhead for maintaining the large graph index. This motivates the following disk-resident graph-based ANNS.

\vspace{0.1cm}
\noindent\textbf{Disk-resident graph-based ANNS.} Most of the solutions are based on hybrid indexing schemes that utilize both disk and memory. GRIP \cite{Zhang2019} reduces the memory consumption of a graph index by compressing high-dimensional vectors to lower-dimensional ones. HM-ANN \cite{Ren2020} addresses memory challenges via heterogeneous storage solutions. SPANN \cite{Chen2021} is a partition tree-based method that maintains a small SPT in memory but the full vectors in disk. BBANN \cite{Simhadri2022} combines a memory resident invert index and on-disk graph index to optimize search tasks. The search accuracy of above methods is affected by the compressed vectors, low-quality invert index, and SPT. So, researchers resort to DiskANN family of methods to reduce memory overhead while ensuring a high accuracy.

\vspace{0.1cm}
\noindent\textbf{DiskANN family of methods.} DiskANN \cite{JayaramSubramanya2019} introduces PQ to assist graph-based ANNS and memory usage. It has been widely deployed in the industry such as Bing search of Microsoft \cite{Zhang2022a} and many follow-up works present variants of DiskANN, e.g., Filter-DiskANN for filtered search \cite{gollapudi2023}, OOD-DiskANN for cross-modal search \cite{Jaiswal2022}, and Fresh-DiskANN for streaming search \cite{Singh2021}. Although DiskANN and its variants achieve a good performance on both memory overhead and search accuracy, they all ignore the fact that frequent SSD I/O would affect the overall QPS. This motivates us to address the I/O issue of DiskANN while retaining its strengths.


\vspace{-0.1cm}
\section{Conclusion}
This paper proposes DiskANN++ that consists of three optimizations for DiskANN. (1) It expedites the convergence of graph search to the query neighborhood by employing query-sensitive entry vertex. (2) It applies isomorphic mapping on Vamana to refine its SSD layout and enhance each SSD page's data value. (3) It replaces original beamsearch by a new pagesearch based on the refined SSD layout with an asynchronous page expansion. Experimental results shows that DiskANN++ achieves a notable 1.5X to 2.2X improvement on QPS compared to DiskANN across various hardware configurations and datasets, given the same accuracy constraint.

\section*{Acknowledgment}
This work was supported by the National NSF of China (62072149), the Primary R\&D Plan of Zhejiang (2023C03198 and 2021C03156), and Fundamental Research Funds for the Provincial Universities of Zhejiang (GK219909299001-006). 


\bibliographystyle{IEEEtran}
\bibliography{references}

\begin{thebibliography}{10}
\providecommand{\url}[1]{#1}
\csname url@samestyle\endcsname
\providecommand{\newblock}{\relax}
\providecommand{\bibinfo}[2]{#2}
\providecommand{\BIBentrySTDinterwordspacing}{\spaceskip=0pt\relax}
\providecommand{\BIBentryALTinterwordstretchfactor}{4}
\providecommand{\BIBentryALTinterwordspacing}{\spaceskip=\fontdimen2\font plus
\BIBentryALTinterwordstretchfactor\fontdimen3\font minus
  \fontdimen4\font\relax}
\providecommand{\BIBforeignlanguage}[2]{{%
\expandafter\ifx\csname l@#1\endcsname\relax
\typeout{** WARNING: IEEEtran.bst: No hyphenation pattern has been}%
\typeout{** loaded for the language `#1'. Using the pattern for}%
\typeout{** the default language instead.}%
\else
\language=\csname l@#1\endcsname
\fi
#2}}
\providecommand{\BIBdecl}{\relax}
\BIBdecl

\bibitem{arya1993approximate}
S.~Arya and D.~M. Mount, ``Approximate nearest neighbor queries in fixed
  dimensions.'' in \emph{SODA}, vol.~93, 1993, pp. 271--280.

\bibitem{indyk1998approximate}
P.~Indyk and R.~Motwani, ``Approximate nearest neighbors: towards removing
  theƒƒ curse of dimensionality,'' in \emph{Proceedings of the thirtieth
  annual ACM symposium on Theory of computing}, 1998, pp. 604--613.

\bibitem{Wang2022}
Q.~Wang, H.~Yin, T.~Chen, J.~Yu, A.~Zhou, and X.~Zhang, ``Fast-adapting and
  privacy-preserving federated recommender system,'' \emph{The VLDB Journal},
  vol.~31, no.~5, pp. 877--896, 2022.

\bibitem{Sarwar2001}
B.~Sarwar, G.~Karypis, J.~Konstan, and J.~Riedl, ``Item-based collaborative
  filtering recommendation algorithms,'' in \emph{Proceedings of the 10th
  international conference on World Wide Web}, 2001, pp. 285--295.

\bibitem{Xu2022}
X.~Xu, J.~Liu, Y.~Wang, and X.~Ke, ``Academic expert finding via $(k,p) $-core
  based embedding over heterogeneous graphs,'' in \emph{2022 IEEE 38th
  International Conference on Data Engineering (ICDE)}.\hskip 1em plus 0.5em
  minus 0.4em\relax IEEE, 2022, pp. 338--351.

\bibitem{Flickner1995}
M.~Flickner, H.~Sawhney, W.~Niblack, J.~Ashley, Q.~Huang, B.~Dom, M.~Gorkani,
  J.~Hafner, D.~Lee, D.~Petkovic \emph{et~al.}, ``Query by image and video
  content: The qbic system,'' \emph{computer}, vol.~28, no.~9, pp. 23--32,
  1995.

\bibitem{Adeniyi2016}
D.~A. Adeniyi, Z.~Wei, and Y.~Yongquan, ``Automated web usage data mining and
  recommendation system using k-nearest neighbor (knn) classification method,''
  \emph{Applied Computing and Informatics}, vol.~12, no.~1, pp. 90--108, 2016.

\bibitem{Bijalwan2014}
V.~Bijalwan, V.~Kumar, P.~Kumari, and J.~Pascual, ``Knn based machine learning
  approach for text and document mining,'' \emph{International Journal of
  Database Theory and Application}, vol.~7, no.~1, pp. 61--70, 2014.

\bibitem{Cover1967}
T.~Cover and P.~Hart, ``Nearest neighbor pattern classification,'' \emph{IEEE
  transactions on information theory}, vol.~13, no.~1, pp. 21--27, 1967.

\bibitem{Zhu2019}
C.~J. Zhu, T.~Zhu, H.~Li, J.~Bi, and M.~Song, ``Accelerating large-scale
  molecular similarity search through exploiting high performance computing,''
  in \emph{2019 IEEE International Conference on Bioinformatics and Biomedicine
  (BIBM)}.\hskip 1em plus 0.5em minus 0.4em\relax IEEE, 2019, pp. 330--333.

\bibitem{Wang2021}
M.~Wang, X.~Xu, Q.~Yue, and Y.~Wang, ``A comprehensive survey and experimental
  comparison of graph-based approximate nearest neighbor search,''
  \emph{Proceedings of the VLDB Endowment}, vol.~14, pp. 1964--1978, 07 2021.

\bibitem{Gionis1999}
A.~Gionis, P.~Indyk, R.~Motwani \emph{et~al.}, ``Similarity search in high
  dimensions via hashing,'' in \emph{Vldb}, vol.~99, no.~6, 1999, pp. 518--529.

\bibitem{Gong2020}
L.~Gong, H.~Wang, M.~Ogihara, and J.~Xu, ``idec: indexable distance estimating
  codes for approximate nearest neighbor search,'' \emph{Proceedings of the
  VLDB Endowment}, vol.~13, no.~9, 2020.

\bibitem{Arora2018}
A.~Arora, S.~Sinha, P.~Kumar, and A.~Bhattacharya, ``Hd-index: Pushing the
  scalability-accuracy boundary for approximate knn search in high-dimensional
  spaces,'' \emph{Proceedings of the VLDB Endowment}, vol.~11, no.~8, 2018.

\bibitem{SilpaAnan2008}
C.~Silpa-Anan and R.~Hartley, ``Optimised kd-trees for fast image descriptor
  matching,'' in \emph{2008 IEEE Conference on Computer Vision and Pattern
  Recognition}.\hskip 1em plus 0.5em minus 0.4em\relax IEEE, 2008, pp. 1--8.

\bibitem{Jegou2010}
H.~Jegou, M.~Douze, and C.~Schmid, ``Product quantization for nearest neighbor
  search,'' \emph{IEEE transactions on pattern analysis and machine
  intelligence}, vol.~33, no.~1, pp. 117--128, 2010.

\bibitem{Chiu2019}
C.-Y. Chiu, A.~Prayoonwong, and Y.-C. Liao, ``Learning to index for nearest
  neighbor search,'' \emph{IEEE transactions on pattern analysis and machine
  intelligence}, vol.~42, no.~8, pp. 1942--1956, 2019.

\bibitem{Fu2017}
C.~Fu, C.~Xiang, C.~Wang, and D.~Cai, ``Fast approximate nearest neighbor
  search with the navigating spreading-out graph,'' \emph{arXiv preprint
  arXiv:1707.00143}, 2017.

\bibitem{Malkov2016}
Y.~Malkov and D.~Yashunin, ``Efficient and robust approximate nearest neighbor
  search using hierarchical navigable small world graphs,'' \emph{IEEE
  Transactions on Pattern Analysis and Machine Intelligence}, vol.~PP, no.~4,
  pp. 824--836, 03 2016.

\bibitem{Chen2021}
Q.~Chen, B.~Zhao, H.~Wang, M.~Li, C.~Liu, Z.~Li, M.~Yang, and J.~Wang, ``Spann:
  Highly-efficient billion-scale approximate nearest neighborhood search,''
  \emph{Advances in Neural Information Processing Systems}, vol.~34, pp.
  5199--5212, 2021.

\bibitem{Aumueller2020}
M.~Aum{\"u}ller, E.~Bernhardsson, and A.~Faithfull, ``Ann-benchmarks: A
  benchmarking tool for approximate nearest neighbor algorithms,''
  \emph{Information Systems}, vol.~87, p. 101374, 2020.

\bibitem{Shimomura2021}
L.~C. Shimomura, R.~S. Oyamada, M.~R. Vieira, and D.~S. Kaster, ``A survey on
  graph-based methods for similarity searches in metric spaces,''
  \emph{Information Systems}, vol.~95, p. 101507, 2021.

\bibitem{Aoyama2013}
K.~Aoyama, A.~Ogawa, T.~Hattori, T.~Hori, and A.~Nakamura, ``Graph index based
  query-by-example search on a large speech data set,'' in \emph{2013 IEEE
  International Conference on Acoustics, Speech and Signal Processing}.\hskip
  1em plus 0.5em minus 0.4em\relax IEEE, 2013, pp. 8520--8524.

\bibitem{Hacid2010}
H.~Hacid and T.~Yoshida, ``Neighborhood graphs for indexing and retrieving
  multi-dimensional data,'' \emph{Journal of Intelligent Information Systems},
  vol.~34, pp. 93--111, 2010.

\bibitem{Paredes2005}
R.~Paredes and E.~Ch{\'a}vez, ``Using the k-nearest neighbor graph for
  proximity searching in metric spaces,'' in \emph{String Processing and
  Information Retrieval: 12th International Conference, SPIRE 2005, Buenos
  Aires, Argentina, November 2-4, 2005. Proceedings 12}.\hskip 1em plus 0.5em
  minus 0.4em\relax Springer, 2005, pp. 127--138.

\bibitem{Simhadri2022}
H.~Simhadri, G.~Williams, M.~Aumüller, M.~Douze, A.~Babenko, D.~Baranchuk,
  Q.~Chen, L.~Hosseini, R.~Krishnaswamy, G.~Srinivasa, S.~Subramanya, and
  J.~Wang, ``Results of the neurips'21 challenge on billion-scale approximate
  nearest neighbor search,'' 05 2022.

\bibitem{JayaramSubramanya2019}
S.~Jayaram~Subramanya, F.~Devvrit, H.~V. Simhadri, R.~Krishnawamy, and
  R.~Kadekodi, ``Diskann: Fast accurate billion-point nearest neighbor search
  on a single node,'' \emph{Advances in Neural Information Processing Systems},
  vol.~32, 2019.

\bibitem{Dobson2023}
M.~Dobson, Z.~Shen, G.~E. Blelloch, L.~Dhulipala, Y.~Gu, H.~V. Simhadri, and
  Y.~Sun, ``Scaling graph-based anns algorithms to billion-size datasets: A
  comparative analysis,'' \emph{arXiv preprint arXiv:2305.04359}, 2023.

\bibitem{Zhang2022a}
J.~Zhang, Z.~Liu, W.~Han, S.~Xiao, R.~Zheng, Y.~Shao, H.~Sun, H.~Zhu,
  P.~Srinivasan, W.~Deng, Q.~Zhang, and X.~Xie, ``Uni-retriever: Towards
  learning the unified embedding based retriever in bing sponsored search,'' in
  \emph{KDD}, 2022, pp. 4493--4501.

\bibitem{gollapudi2023}
S.~Gollapudi, N.~Karia, V.~Sivashankar, R.~Krishnaswamy, N.~Begwani, S.~Raz,
  Y.~Lin, Y.~Zhang, N.~Mahapatro, P.~Srinivasan, A.~Singh, and H.~V. Simhadri,
  ``Filtered-diskann: Graph algorithms for approximate nearest neighbor search
  with filters,'' in \emph{Proceedings of the ACM Web Conference 2023}, 2023,
  p. 3406–3416.

\bibitem{Jaiswal2022}
S.~Jaiswal, R.~Krishnaswamy, A.~Garg, H.~V. Simhadri, and S.~Agrawal,
  ``Ood-diskann: Efficient and scalable graph anns for out-of-distribution
  queries,'' \emph{arXiv preprint arXiv:2211.12850}, 2022.

\bibitem{Singh2021}
A.~Singh, S.~J. Subramanya, R.~Krishnaswamy, and H.~V. Simhadri,
  ``Freshdiskann: A fast and accurate graph-based ann index for streaming
  similarity search,'' \emph{arXiv preprint arXiv:2105.09613}, 2021.

\bibitem{Xu2021}
X.~Xu, M.~Wang, Y.~Wang, and D.~Ma, ``Two-stage routing with optimized guided
  search and greedy algorithm on proximity graph,'' \emph{Knowledge-Based
  Systems}, vol. 229, p. 107305, 07 2021.

\bibitem{Dearholt1988}
D.~Dearholt, N.~Gonzales, and G.~Kurup, ``Monotonic search networks for
  computer vision databases,'' vol.~2, 02 1988, pp. 548--553.

\bibitem{Zhu2021}
D.~Zhu and M.~Zhang, ``Understanding and generalizing monotonic proximity
  graphs for approximate nearest neighbor search,'' \emph{arXiv preprint
  arXiv:2107.13052}, 2021.

\bibitem{Babenko2011}
M.~Babenko and A.~Gusakov, ``New exact and approximation algorithms for the
  star packing problem in undirected graphs,'' \emph{Symposium on Theoretical
  Aspects of Computer Science (STACS2011)}, vol.~9, 03 2011.

\bibitem{Peng2018}
K.~Peng and Q.~Huang, ``Clustering approach based on mini batch kmeans for
  intrusion detection system over big data,'' \emph{IEEE Access}, vol.~PP, pp.
  1--1, 02 2018.

\bibitem{Fiedler1973}
\BIBentryALTinterwordspacing
M.~Fiedler, ``Algebraic connectivity of graphs,'' \emph{Czechoslovak
  Mathematical Journal}, vol.~23, pp. 298--305, 1973. [Online]. Available:
  \url{https://api.semanticscholar.org/CorpusID:117770486}
\BIBentrySTDinterwordspacing

\bibitem{Zhang2017}
C.~Zhang, F.~Wei, Q.~Liu, Z.~Tang, and Z.~Li, ``Graph edge partitioning via
  neighborhood heuristic,'' 08 2017, pp. 605--614.

\bibitem{Xie2014}
C.~Xie, L.~Yan, W.-J. Li, and Z.~Zhang, ``Distributed power-law graph
  computing: Theoretical and empirical analysis,'' \emph{Advances in Neural
  Information Processing Systems}, vol.~2, pp. 1673--1681, 01 2014.

\bibitem{Petroni2015}
F.~Petroni, L.~Querzoni, K.~Daudjee, S.~Kamali, and G.~Iacoboni, ``Hdrf:
  Stream-based partitioning for power-law graphs,'' 10 2015, pp. 243--252.

\bibitem{Karypis1970}
G.~Karypis, V.~Kumar, and S.~Comput, ``A fast and high quality multilevel
  scheme for partitioning irregular graphs,'' \emph{SIAM Journal on Scientific
  Computing}, vol.~20, 02 1970.

\bibitem{Stanton2012}
I.~Stanton and G.~Kliot, ``Streaming graph partitioning for large distributed
  graphs,'' 09 2012.

\bibitem{Tsourakakis2014}
C.~Tsourakakis, C.~Gkantsidis, B.~Radunovic, and M.~Vojnovic, ``Fennel:
  Streaming graph partitioning for massive scale graphs,'' 02 2014, pp.
  333--342.

\bibitem{Wei2016}
H.~Wei, J.~Yu, C.~Lu, and X.~Lin, ``Speedup graph processing by graph
  ordering,'' 06 2016, pp. 1813--1828.

\bibitem{Facco2017}
E.~Facco, M.~d’Errico, A.~Rodriguez, and A.~Laio, ``Estimating the intrinsic
  dimension of datasets by a minimal neighborhood information,''
  \emph{Scientific Reports}, vol.~7, 09 2017.

\bibitem{texmex}
\BIBentryALTinterwordspacing
Anon. (2010) Datasets for approximate nearest neighbor search. [Online].
  Available: \url{http://corpus-texmex.irisa.fr/}
\BIBentrySTDinterwordspacing

\bibitem{Artem2016}
A.~Yandex and V.~Lempitsky, ``Efficient indexing of billion-scale datasets of
  deep descriptors,'' 06 2016, pp. 2055--2063.

\bibitem{Msong}
\BIBentryALTinterwordspacing
Song. (2011) Million song dataset benchmarks. [Online]. Available:
  \url{http://www.ifs.tuwien.ac.at/mir/msd/}
\BIBentrySTDinterwordspacing

\bibitem{Crawl}
\BIBentryALTinterwordspacing
Crawl. (2023) Common crawl. [Online]. Available: \url{http://commoncrawl.org/}
\BIBentrySTDinterwordspacing

\bibitem{turing}
H.~Zhang, X.~Song, C.~Xiong, C.~Rosset, P.~Bennett, N.~Craswell, and S.~Tiwary,
  ``Generic intent representation in web search,'' 07 2019, pp. 65--74.

\bibitem{Glove}
\BIBentryALTinterwordspacing
P.~Jeffrey, S.~Richard, and D.~M. Christopher. (2015) Glove: Global vectors for
  word representation. [Online]. Available:
  \url{http://nlp.stanford.edu/projects/glove/}
\BIBentrySTDinterwordspacing

\bibitem{SPANNrepo}
microsoft, ``Sptag: A library for fast approximate nearest neighbor search,''
  \url{https://github.com/microsoft/SPTAG}, 2023.

\bibitem{BBANNrepo}
Zilliztech, ``Bbann: Block-based approximate nearest neighbor,''
  \url{https://github.com/zilliztech/BBAnn}, 2023.

\bibitem{DiskANNrepo}
microsoft, ``Diskann,'' \url{https://github.com/microsoft/DiskANN}, 2023.

\bibitem{Li2019}
W.~Li, Y.~Zhang, Y.~Sun, W.~Wang, M.~Li, W.~Zhang, and X.~Lin, ``Approximate
  nearest neighbor search on high dimensional data—experiments, analyses, and
  improvement,'' \emph{IEEE Transactions on Knowledge and Data Engineering},
  vol.~32, no.~8, pp. 1475--1488, 2019.

\bibitem{Jegou2011}
H.~Jégou, M.~Douze, and C.~Schmid, ``Product quantization for nearest neighbor
  search,'' \emph{IEEE transactions on pattern analysis and machine
  intelligence}, vol.~33, pp. 117--28, 01 2011.

\bibitem{Ge2013}
T.~Ge, K.~He, Q.~Ke, and J.~Sun, ``Optimized product quantization for
  approximate nearest neighbor search,'' in \emph{Proceedings of the IEEE
  Conference on Computer Vision and Pattern Recognition}, 2013, pp. 2946--2953.

\bibitem{Kalantidis2014}
Y.~Kalantidis and Y.~Avrithis, ``Locally optimized product quantization for
  approximate nearest neighbor search,'' \emph{In CVPR}, pp. 2321--2328, 01
  2014.

\bibitem{Wang2017}
J.~Wang and T.~Zhang, ``Composite quantization,'' \emph{IEEE Transactions on
  Pattern Analysis and Machine Intelligence}, vol.~PP, 12 2017.

\bibitem{Pan2020}
Z.~Pan, L.~Wang, Y.~Wang, and Y.~Liu, ``Product quantization with dual
  codebooks for approximate nearest neighbor search,'' \emph{Neurocomputing},
  vol. 401, 03 2020.

\bibitem{Weiss2008}
Y.~Weiss, A.~Torralba, and R.~Fergus, ``Spectral hashing,'' in \emph{Advances
  in Neural Information Processing Systems}, D.~Koller, D.~Schuurmans,
  Y.~Bengio, and L.~Bottou, Eds., vol.~21.\hskip 1em plus 0.5em minus
  0.4em\relax Curran Associates, Inc., 2008.

\bibitem{Xu2011}
H.~Xu, J.~Wang, Z.~Li, G.~Zeng, S.~Li, and N.~Yu, ``Complementary hashing for
  approximate nearest neighbor search,'' 11 2011, pp. 1631--1638.

\bibitem{Huang2017}
Q.~Huang, J.~Feng, Q.~Fang, W.~Ng, and W.~Wang, ``Query-aware
  locality-sensitive hashing scheme for lp norm,'' \emph{The VLDB Journal},
  vol.~26, no.~5, pp. 683--708, 2017.

\bibitem{Karaman2019}
S.~Karaman, X.~Lin, X.~Hu, and S.-F. Chang, ``Unsupervised rank-preserving
  hashing for large-scale image retrieval,'' in \emph{Proceedings of the 2019
  on International Conference on Multimedia Retrieval}, 2019, pp. 192--196.

\bibitem{Hajebi2011}
K.~Hajebi, Y.~Abbasi-Yadkori, H.~Shahbazi, and H.~Zhang, ``Fast approximate
  nearest-neighbor search with k-nearest neighbor graph,'' in
  \emph{Twenty-Second International Joint Conference on Artificial
  Intelligence}, 2011.

\bibitem{Malkov2014}
Y.~Malkov, A.~Ponomarenko, A.~Logvinov, and V.~Krylov, ``Approximate nearest
  neighbor algorithm based on navigable small world graphs,'' \emph{Information
  Systems}, vol.~45, pp. 61--68, 2014.

\bibitem{Fu2016}
C.~Fu and D.~Cai, ``Efanna : An extremely fast approximate nearest neighbor
  search algorithm based on knn graph,'' 09 2016.

\bibitem{Zhang2019}
M.~Zhang and Y.~He, ``Grip: Multi-store capacity-optimized high-performance
  nearest neighbor search for vector search engine,'' in \emph{Proceedings of
  the 28th ACM International Conference on Information and Knowledge
  Management}, 11 2019, pp. 1673--1682.

\bibitem{Ren2020}
J.~Ren, M.~Zhang, and D.~Li, ``Hm-ann: Efficient billion-point nearest neighbor
  search on heterogeneous memory,'' \emph{Advances in Neural Information
  Processing Systems}, vol.~33, pp. 10\,672--10\,684, 2020.

\end{thebibliography}

\end{document}